\documentclass[12pt, a4paper]{article}

\usepackage{indentfirst}
\usepackage{graphicx, amsmath,amsthm,amsfonts, amssymb}
\usepackage{color, soul, comment}
\usepackage{authblk}
\usepackage{mathtools}
\usepackage{tikz}
\usepackage{enumerate, caption, float}
\usetikzlibrary{calc,intersections,through,backgrounds}
\usepackage{enumerate}
\usepackage{chngcntr}
\usepackage[hidelinks]{hyperref}
\usepackage[latin10,utf8]{inputenc}

\newtheorem{theorem}{Theorem}
\newtheorem{corollary}[theorem]{Corollary}
\newtheorem{lemma}[theorem]{Lemma}
\newtheorem{proposition}[theorem]{Proposition}

\newtheorem{definition}[theorem]{Definition}

\title{One-Switch Discount Functions}
\author{Nina Anchugina}
\affil{\small Department of Mathematics, University of Auckland}
\affil[ ]{\small \tt n.anchugina@auckland.ac.nz}
\date{February 2017}
\begin{document}

	\maketitle
		
	\bigskip
	\noindent 
	\par\bigskip
	
	\noindent
	{\bf Abstract.}  Bell \cite{bell1988one} introduced the one-switch property for preferences over sequences of dated outcomes. This property concerns the effect of adding a common delay to two such sequences: it says that the preference ranking of the delayed sequences is either independent of the delay, or else there is a unique delay such that one strict ranking prevails for shorter delays and the opposite strict ranking for longer delays. 
	For preferences that have a discounted utility (DU) representation, Bell \cite{bell1988one} argues that the only discount functions consistent with the one-switch property are sums of exponentials. This paper proves that discount functions of the linear times exponential form also satisfy the one-switch property. We further demonstrate that preferences which have a DU representation with a linear times exponential discount function exhibit increasing impatience (\cite{takeuchi2011non}).  We also clarify an ambiguity in the original Bell \cite{bell1988one} definition of the one-switch property by distinguishing a weak one-switch property from the (strong) one-switch property. We show that the one-switch property and the weak one-switch property definitions are equivalent in a continuous-time version of the Anscombe and Aumann \cite{anscombe1963definition} setting.
	
	\par\bigskip\bigskip
	
	
	\bigskip \bigskip \bigskip
	\vfill
	\noindent
	\setcounter{page}{0}
	\thispagestyle{empty}
	\newpage



\bigskip

In this paper we analyse the one-switch property for intertemporal preferences introduced by Bell \cite{bell1988one}. The one-switch property was initially formulated for preferences over lotteries \cite{bell1988one}. 
It says that the preference ranking of any two lotteries is either independent of wealth, or else there is a unique level of wealth such that one strict ranking prevails for lower wealth levels and the opposite strict ranking for higher wealth levels.
For preferences with an expected utility representation, the one-switch property restricts the form of the Bernoulli utility function.
As demonstrated by Bell \cite{bell1988one}, the utility functions that satisfy this property are the quadratic, the sum of exponentials, the linear plus exponential and the linear times exponential. The properties of these functions, and their possible applications, have been extensively investigated in risk theory
(see, for example, \cite{abbas2011one}, \cite{abbas2012methods}, \cite{bell2000utility} and \cite{bell2001strong}). 
However, it is less well known that Bell \cite{bell1988one} also defined an analogous one-switch property for preferences over sequences of dated outcomes. 
In this case, the one-switch property concerns the effect of adding a common delay to two sequences of dated outcomes: it says that the preference ranking of the delayed sequences is either independent of the delay, or else there is a unique delay such that one strict ranking prevails for shorter delays and the opposite strict ranking for longer delays.
Bell \cite{bell1988one} claims that if preferences have a discounted utility representation, then the only discount functions consistent with the one-switch property are sums of exponentials. 

However, in this paper we show that the one-switch property is also compatible with another form of discount function: the linear times exponential. To the best of our knowledge, this type of discount function has not been previously used in intertemporal context. As we demonstrate in the paper, linear times exponential discount functions exhibit strictly increasing impatience (II). While strictly II has not been a very frequent experimental observation \cite{frederick2002time}, some recent empirical findings support this type of impatience \cite{bleichrodt2016measurement, sayman2009investigation, takeuchi2011non}. We also analyse the distinction between the weak one-switch property and the (strong) one-switch property in the time preference context. This distinction is based on whether weak or strict preferences are reversed at the switch point. While this distinction is inconsequential in the risk set-up with an expected utility representation, matters are not so clear for the intertemporal context, even assuming a discounted utility representation. We show, however, that these
two properties are equivalent in a set-up analogous to that of Anscombe and Aumann \cite{anscombe1963definition}. 

The paper is organized as follows. We start by giving some preliminaries in Section \ref{section:prelimOSP}. 
Section \ref{section:OSP} is devoted to revising Bell's characterisation \cite[Proposition 2]{bell1988one} of the discount functions that exhibit the one-switch property. We first discuss an ambiguity in Bell's \cite{bell1988one} definition of this property, and distinguish a standard (strong) version from an alternative ``weak" one-switch property. We show that the discount functions consistent with the (standard) one-switch property are those which have the sum of exponentials or the linear time exponential form. We also explore the relationship between the one-switch property restricted to preferences over single dated outcomes and the impatience properties of such preferences.
In Section \ref{section:w1s} we study the weak one-switch property. In the context of expected utility preferences over lotteries, where the one-switch property refers to the effect of wealth level on the ranking of two lotteries, the results in \cite{bell1988one} imply that the weak one-switch property is equivalent to the standard one. In the intertemporal context, we establish that the equivalence also holds if we endow $X$ with a mixture set structure and work in an environment similar to of Anscombe and Aumann in \cite{anscombe1963definition}. Finally, Section \ref{section:discOSP} summarizes the results.

\section{Preliminaries}\label{section:prelimOSP}
Consider preferences over sequences of dated outcomes. We work in a continuous time environment throughout this paper. Points in time are elements of the set $T = [0, \infty)$, where the present time corresponds to $t=0$. The set of outcomes is initially assumed to be the interval $X = [0, \infty)$, though we will re-define $X$ to be an arbitrary mixture set in Section \ref{section:w1s}. 

Let ${\mathcal A}_n=\{ \ ({\bf x},{\bf t}) \in X^n\times T^n \ \vert \ t_1<t_2<\ldots<t_n \ \}$ be the set of sequences with $n$ dated outcomes.
Define the set of alternatives ${\mathcal A}$ as follows: ${\mathcal  A}=\cup_{n=1}^{\infty}{\mathcal A}_n.$ The set ${\mathcal A}$ consists of all sequences of finitely many dated outcomes. Elements of ${\mathcal A_1}\subseteq{\mathcal  A}$ are called dated outcomes.

Consider a preference order $\succcurlyeq$ on the set of alternatives ${\mathcal A}$. 

We say that $U$ is a {\em discounted utility (DU) representation} for $\succcurlyeq$, if $U$ represents $\succcurlyeq$ and there exist $(u, D)$, such that  $u \colon X \to \mathbb{R}$ is a utility function (continuous, strictly increasing with $u(0)=0$), $D \colon T \to (0, 1]$ is a discount function (strictly decreasing, $D(0)=1$ and $\lim_{t\to \infty}D(t)=0$) and 
\[
U({\bf x, t}) = \sum_{i=1}^{n} D(t_i)u(x_i)
\]
for all $n$ and every $({\bf x, t}) \in {\mathcal A}_n$. 
Necessary and sufficient conditions for a DU representation were provided by Harvey \cite[Theorem 2.1]{harvey1995proportional}. 

We denote the set of positive integers $\{1, 2, 3, \ldots\}$ by $\mathbb{N}$, and the set of non-negative integers $\{0, 1, 2, 3, \ldots\}$ by $\mathbb{N}_0$, so $\mathbb{N} \subset \mathbb{N}_0$.

\section{The one-switch property}\label{section:OSP}
\subsection{One-switch discount functions}\label{section:OSDiscountFunctions}

For any sequence of dated outcomes $({\bf x}, {\bf t}) \in {\mathcal A}_n$ and any delay $\sigma>0$, let 
\[
	({\bf x}, {\bf t}+\sigma)=({\bf x}, (t_1+\sigma, \ldots, t_n+\sigma))
\]
denote the delayed sequence.
\begin{definition}[\cite{bell1988one}] \label{1s}
	We say that the preferences $\succcurlyeq$ on ${\mathcal A}$ exhibit the {\em one-switch property} if for every pair $({\bf x, t}), ({\bf y, s}) \in {\mathcal A}$ the ranking of $({\bf x}, {\bf t}+\sigma)$ and $({\bf y}, {\bf s}+\sigma)$ is either independent of $\sigma$, or there exists $\sigma^*\geq 0$, such that  
	\begin{align*}
		({\bf x}, {\bf t}+\sigma) &\succ ({\bf y}, {\bf s}+\sigma), \text{ for any } \sigma<\sigma^*,\\
		({\bf x}, {\bf t}+\sigma) &\prec ({\bf y}, {\bf s}+\sigma), \text{ for any } \sigma>\sigma^*
	\end{align*}
	or
	\begin{align*}
		({\bf x}, {\bf t}+\sigma) &\prec({\bf y}, {\bf s}+\sigma), \text{ for any } \sigma<\sigma^*,\\
		({\bf x}, {\bf t}+\sigma) &\succ ({\bf y}, {\bf s}+\sigma), \text{ for any } \sigma>\sigma^*.
	\end{align*}
\end{definition}
It is worth mentioning that Bell's \cite{bell1988one} original verbal definition of the one-switch property uses the ambiguous word ``preferred'', which does not specify whether the preference order is used in a strong or weak sense. Bell's Lemma 3 \cite{bell1988one} implicitly suggests that weak preference is intended, but this Lemma also shows that either interpretation leads to the same restriction on expected utility preferences. Abbas and Bell \cite{abbas2015ordinal} introduce a formal definition which is explicit about preference being strict. Therefore, we define the one-switch property using strict preference order. 

The one-switch property can be stated in a weaker variant, as follows:
\begin{definition}[\cite{bell1988one}]
	We say that the preferences $\succcurlyeq$ on ${\mathcal A}$ exhibit the {\em weak one-switch property} if for every pair $({\bf x, t}), ({\bf y, s}) \in {\mathcal A}$ the ranking of $({\bf x}, {\bf t}+\sigma)$ and $({\bf y}, {\bf s}+\sigma)$ is either independent of $\sigma$, or there exists $\sigma^*\geq 0$, such that  
	\begin{align*}
		({\bf x}, {\bf t}+\sigma) &\succcurlyeq ({\bf y}, {\bf s}+\sigma), \text{ for any } \sigma<\sigma^*,\\
		({\bf x}, {\bf t}+\sigma) &\preccurlyeq ({\bf y}, {\bf s}+\sigma), \text{ for any } \sigma>\sigma^*
	\end{align*}
	or
	\begin{align*}
		({\bf x}, {\bf t}+\sigma) &\preccurlyeq({\bf y}, {\bf s}+\sigma), \text{ for any } \sigma<\sigma^*,\\
		({\bf x}, {\bf t}+\sigma) &\succcurlyeq ({\bf y}, {\bf s}+\sigma), \text{ for any } \sigma>\sigma^*.
	\end{align*}
\end{definition}
In other words, there do not exist $({\bf x, t}), ({\bf y, s}) \in {\mathcal A}$ and $\sigma, \varepsilon$ with $0<\sigma< \varepsilon$ such that 
\begin{align*}
	({\bf x}, {\bf t}) &\succ ({\bf y}, {\bf s}), \\
	({\bf x}, {\bf t}+\sigma) &\prec ({\bf y}, {\bf s}+\sigma), \\
	({\bf x}, {\bf t}+\varepsilon) &\succ ({\bf y}, {\bf s}+\varepsilon),
\end{align*}
or with all strict preferences reversed.

In the intertemporal context it is not known whether this alternative ``weak" version is equivalent (given a DU representation) to Definition \ref{1s}. This question will be investigated is Section \ref{section:w1s}, where we adapt Bell's Lemma 3 \cite{bell1988one} to the temporal setting. We demonstrate that the one-switch property and the weak one-switch property are equivalent in an intertemporal version of the Anscombe and Aumann (AA) \cite{anscombe1963definition} environment similar to that investigated in \cite{anchugina2016simple}.

If preferences $\succcurlyeq$ on ${\mathcal A}$ have a DU representation $(u, D)$, then the one-switch property means that for any $({\bf x, t}), ({\bf y, s}) \in {\mathcal A}$ the function
\[
\Delta(\sigma)= \sum_{i=1}^nD(t_i+\sigma)u(x_i)-\sum_{j=1}^mD(s_j+\sigma)u(y_j)
\] 
either has constant sign or else there is some $\sigma^{*}\geq 0$ such that $\Delta(\sigma)=0$ if and only if $\sigma=\sigma^{*}$ and $\Delta(\sigma')\Delta(\sigma'')>0$ if and only if $\sigma'\neq \sigma^{*}$ and $\sigma ''\neq \sigma^{*}$ are on the same side of $\sigma^{*}$.
That is,
\begin{align}\label{eqn:signcondition}
	\begin{split}
		&sign\left(\Delta(\sigma)\right)=const \text{ for all }\sigma\geq 0, \text{ or else} \\
		&\text{there exists }\sigma^{*} \text{ such that } \Delta(\sigma^*)=0 \text{ and }\\
		&\Delta(\sigma')\Delta(\sigma'')>0 \text{ if }\sigma', \sigma''>\sigma^*\text{ or }\sigma', \sigma''<\sigma^* \text{ and }\\
		&\Delta(\sigma')\Delta(\sigma'')<0 \text{ if } \sigma'<\sigma^*<\sigma'' \text{ or }\sigma''<\sigma^*<\sigma'.
	\end{split}
\end{align}
Figure \ref{fig:difference} provides an illustration of $\Delta(\sigma)$ for preferences which exhibit the one-switch property and have a DU representation. Note that only the sign of $\Delta(\sigma)$ is relevant.
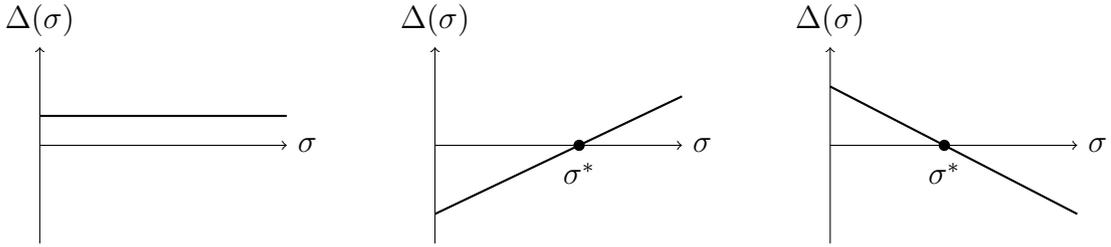
\begin{figure}[!htb] 
	\begin{minipage}{\linewidth}
		\centering
		\begin{tikzpicture}[
		scale=1.3,
		every node/.style={color=black},
		dot/.style={circle,fill=black,minimum size=4pt,inner sep=0pt,
			outer sep=-1pt},
		]
		\draw[name path=xline, ->] (0, 0)--(2.5,0) node(xline)[right] {$\sigma$}; 
		\draw[->] (0,-1)--(0,1) node(yline)[above] {$\Delta(\sigma)$};
		\draw[thick,name path=line,-] (0,0.3)--(2.5,0.3);
		\begin{scope}[xshift=4cm]
		\draw[name path=xline, ->] (0, 0)--(2.5,0) node(xline)[right] {$\sigma$}; 
		\draw[->] (0,-1)--(0,1) node(yline)[above] {$\Delta(\sigma)$};
		\draw[thick, name path=line,-] (0,-0.7)--(2.5,0.5);
		\path [name intersections={of=xline and line,by=E1}];
		\node [circle, fill=black,inner sep=1.5pt,label=-90:$\sigma^*$] at (E1) {};
		\end{scope}
		\begin{scope}[xshift=8cm]
		\draw[name path=xline, ->] (0, 0)--(2.5,0) node(xline)[right] {$\sigma$}; 
		\draw[->] (0,-1)--(0,1) node(yline)[above] {$\Delta(\sigma)$};
		\draw[thick,name path=line,-] (0,0.6)--(2.5,-0.7);
		\path [name intersections={of=xline and line,by=E2}];
		\node [circle, fill=black,inner sep=1.5pt,label=-90:$\sigma^*$] at (E2) {};
		\end{scope}
		\end{tikzpicture}
	\end{minipage}\\
	\caption{Possible behaviour of $\Delta(\sigma)$ for the preferences which exhibit the one-switch property and have a DU representation}
	\label{fig:difference}
\end{figure}


Note that the assumed properties of $u$ imply $u(X)= [0, \bar{u}]$ for some $\bar{u}>0$ or $u(X) = \mathbb{R}_+$. We say that a discount function $D$ satisfies the {\em extended} one-switch property if the function $\Delta \colon \mathbb{R}_+ \to \mathbb{R}$ defined by
\begin{equation} \label{eqn:strongOSP}
	\Delta(\sigma)= \sum_{i=1}^n D(t_i+\sigma)v_i-\sum_{j=1}^m D(s_j+\sigma)w_j
\end{equation}
satisfies \eqref{eqn:signcondition} for any $n, m$, any ${\bf t} \in T^n, {\bf s} \in T^m$ and any ${\bf v} \in \mathbb{R}_+^n, {\bf w} \in \mathbb{R}_+^m$.

\begin{lemma}\label{lemma:OSPonlyD}
	$D$ satisfies the extended one-switch property if and only if there exists $\bar{u}>0$ such that \eqref{eqn:strongOSP} satisfies \eqref{eqn:signcondition} for any $n, m$, any ${\bf t} \in T^n, {\bf s} \in T^m$ and any ${\bf v} \in [0, \bar{u}]^n, {\bf w} \in [0, \bar{u}]^m$.
\end{lemma}
\begin{proof}
	{\em ``Only If".} This part is straightforward.
	
	{\em ``If".} Suppose that there exists $\bar{u}>0$ such that $\Delta:\mathbb{R}_{+}\rightarrow \mathbb{R}$ defined by \eqref{eqn:strongOSP} satisfies \eqref{eqn:signcondition} for any $n, m$, any ${\bf t} \in T^n, {\bf s} \in T^m$ and any ${\bf v} \in [0, \bar{u}]^n, {\bf w} \in [0, \bar{u}]^m$. The proof is by contradiction. Assume that there is some $n', m'$, ${\bf t'} \in T^{n'}, {\bf s'} \in T^{m'}$ and some ${\bf v'} \in \mathbb{R}_+^{n'}, {\bf w'} \in \mathbb{R}_+^{m'}$ such that the function $\Delta^* \colon \mathbb{R}_+ \to \mathbb{R}$ defined by
	\begin{equation*}
		\Delta^*(\sigma)=\sum_{i=1}^{n'} D(t'_i+\sigma) v'_i-\sum_{j=1}^{m'} D(s'_j+\sigma) w'_j.
	\end{equation*}
	violates property \eqref{eqn:signcondition}. Then the function $\lambda \Delta^*$ will also violate \eqref{eqn:signcondition} for any $\lambda>0$.
	Let $\lambda \in (0, 1)$ be such that $\lambda {\bf v'} \in [0, \bar{u}]^{n'}$ and $\lambda {\bf w'} \in [0, \bar{u}]^{m'}$. 
	This is a contradiction to the initial assumption that \eqref{eqn:strongOSP} satisfies \eqref{eqn:signcondition} for any $n, m$, any ${\bf t} \in T^n, {\bf s} \in T^m$ and any ${\bf v} \in [0, \bar{u}]^n, {\bf w} \in [0, \bar{u}]^m$.
\end{proof}

In other words, Lemma \ref{lemma:OSPonlyD} states that given preferences with a DU representation, the range of $u$ is irrelevant to whether or not the preferences satisfy the one-switch property. It follows that the one-switch property does not impose any additional restrictions on the shape of $u$. In other words, given a DU representation $(u, D)$ for $\succcurlyeq$, the one-switch property restricts only $D$. We therefore say that a discount function, $D$, exhibits the one-switch property if there is some utility function, $u$, such that the preferences with DU representation $(u,D)$ exhibit the one-switch property.
Bell \cite[Proposition 8]{bell1988one} argues that sums of exponentials are the only discount functions compatible with the one-switch property. However, we will demonstrate in this section that linear times exponential discount functions also have the one-switch property. 

The following proposition gives the restrictions on the parameters of linear times exponential and sum of exponential functions under which they satisfy the properties of a discount function.
\begin{proposition}\label{discfun}
	\begin{enumerate}[(a)]
		\item The linear times exponential function $D(t) = (c_1 + c_2t)e^{r_1t}$ is a discount function if and only if $c_1=1, -r_1 \geq c_2 \geq 0 \text{ and } r_1<0$; i.e., $D(t) = (1 + c t)e^{-rt}$, where $r \geq c \geq 0$ and $r>0$.
		\item The sum of exponentials function $D(t) = c_1e^{r_1t} + c_2e^{r_2t}$, where $r_1\neq r_2$, is a discount function if and only if
		\begin{itemize} 
			\item $c_2=1-c_1, r_1<r_2<0$ and $\frac{r_2}{r_2-r_1}\leq c_1<1$, or 
			\item $c_2=1-c_1, r_2<r_1<0$ and $0<c_1\leq \frac{r_2}{r_2-r_1}$.
		\end{itemize}
		Equivalently, $D(t) = ae^{-bt} + (1-a)e^{-(b+c)t}$, where $a, b, c>0$, $a\leq b/c+1$.
	\end{enumerate}
\end{proposition}
\begin{proof}
	We need to find the parameters of linear times exponential functions and sums of exponentials such that the following four properties are satisfied: 
	\begin{enumerate}
		\item $D(0)=1$, 
		\item $D(t)>0$ for all $t$, 
		\item $D(t)$ is strictly decreasing, and 
		\item $\lim_{t \to \infty}D(t)=0$.
	\end{enumerate}
	{\em (a) Linear times exponential.}
	Assume that $D(t) = (c_1 + c_2t)e^{r_1t}$. 
	Obviously, $D(0)=1$ if and only if $c_1=1$.
	
	Next, to satisfy the first and the second conditions simultaneously, that is, to have $D(0)=1$ and $D(t)>0$ for all $t$, it is necessary and sufficient that $c_1=1$ and $c_2\geq 0$ (since $e^{r_1t}>0$ for all $t$).
	
	To check whether $D(t)$ is strictly decreasing, consider its first order derivative:
	\[
	D'(t) \ =\ c_2e^{r_1t}+(1+c_2t)r_1e^{r_1t} \ =\ e^{r_1t}\left(c_2+r_1+c_2r_1t\right).
	\]
	Since $e^{rt}>0$ for all $t$, the sign of the derivative depends on the sign of the linear expression $c_2+r_1+c_2r_1t$.  
	Therefore, $D(t)$ is strictly decreasing if and only if 
	\begin{align*}
		c_2+r_1 &\leq 0, \text{ and} \\
		c_2+r_1+c_2r_1t&<0 \text{ for all } t>0.
	\end{align*}
	This condition is equivalent to requiring that one of the following two conditions is satisfied: 
	\begin{itemize}
		\item $c_2+r_1=0$ and $c_2r_1<0$,
		\item $c_2+r_1<0$ and $c_2r_1\leq 0$.
	\end{itemize}
	Since $c_2 \geq 0$, in the first case we have $c_2=-r_1>0$. In the second case $0 \leq c_2<-r_1$. We can summarize both cases as follows:
	\[
	-r_1 \geq c_2 \geq 0 \text{ and } r_1<0. 
	\]
	Therefore, the first three properties of discount functions are satisfied if and only if $c_1=1, -r_1 \geq c_2 \geq 0 \text{ and } r_1<0$. 
	
	Finally, the limit of the linear times exponential function is
	\[
	\displaystyle \lim_{t \to \infty}D(t) \ =\ \displaystyle \lim_{t \to \infty}\left(1 + c_2 t\right)e^{r_1t} \ =\ \displaystyle \lim_{t \to \infty}\frac{1 + c_2t}{e^{-r_1t}},
	\] 
	If $r_1=c_2=0$ this limit is 1. This case is ruled out since from the previous step $r_1<0$. Then, by L'Hopital's rule we have
	\[
	\lim_{t \to \infty}D(t) \ =\ \lim_{t \to \infty}\frac{1 + c_2t}{e^{-r_1t}} \ =\ \lim_{t \to \infty} \frac{c_2}{-r_1e^{-r_1t}} \ =\ 0.
	\] 
	Therefore, $D(t) = (c_1 + c_2t)e^{r_1t}$ satisfies all four properties of discount functions if and only if $c_1=1, -r_1 \geq c_2 \geq 0 \text{ and } r_1<0$. Denote $c=c_2$ and $r=-r_1$. Then we have  $D(t) = (1 + c t)e^{-rt}$, where $r \geq c \geq 0$ and $r > 0$.
	
	{\em (b) Sums of exponentials.} The proof is analogous to \cite[Proposition 8]{bell1988one} and is given here for completeness. Assume that $D(t) = c_1e^{r_1t} + c_2e^{r_2t}$ with $r_1 \neq r_2$.
	
	The condition $D(0)=1$ is satisfied if and only if $c_1+c_2=1$.
	
	We must also have $D(t)>0$ for all $t>0$. Note that 
	\[
	D(t) \ =\ c_1e^{r_1t} + c_2e^{r_2t} \ =\ e^{r_1t}\left(c_1 + c_2e^{(r_2-r_1)t}\right)>0 \text{ for all } t>0 
	\]
	if and only if $c_1 + c_2e^{(r_2-r_1)t}>0$ for all $t>0$ (since $e^{r_1t}>0$ for all $t>0$). From the first condition we know that $c_2=1-c_1$. Substituting this expression to the inequality we must have $c_1 + (1-c_1)e^{(r_2-r_1)t}>0$ for all $t>0$. Therefore, the first two properties of discount functions are satisfied if and only if $c_2=1-c_1$ and one of the following two conditions holds: 
	\begin{enumerate}[(i)]
		\item $r_{1}<r_{2}$ and $c_{1}\leq 1$; or 
		\item $r_{1}>r_{2}$ and $c_{1}\geq 0$.
	\end{enumerate}
	Next, it is necessary to have $D(t)$ strictly decreasing. Consider its first order derivative
	\[
	D'(t) \ =\ c_1r_1e^{r_1t}+c_2r_2e^{r_2t} \ =\ e^{r_1t}\left(c_1r_1+c_2r_2e^{(r_2-r_1)t}\right).
	\]
	Since $e^{r_1t}>0$ for all $t>0$, the function $D(t)$ is strictly decreasing if and only if 
	\begin{align*}
		c_1r_1+c_2r_2 &\leq 0, \text{ and } \\
		c_1r_1+c_2r_2e^{(r_2-r_1)t}&<0 \text{ for all } t>0.
	\end{align*}
	Recall that from the first two conditions we have $c_2=1-c_1$ and [(i) or (ii)] holds.
	Therefore, $D(0)=1, D(t)>0$ for all $t>0$ and $D(t)$ is strictly decreasing if and only if all of the following conditions hold:
	\begin{align*}
		&c_1r_1+(1-c_1)r_2 \leq 0,\\
		&c_1r_1+(1-c_1)r_2e^{(r_2-r_1)t}<0 \text{ for all } t>0,\\
		&c_2=1-c_1, \text{ and   [(i) or (ii)] holds}.
	\end{align*}
	Note that 
	\begin{align*}
		c_1r_1+(1-c_1)r_2e^{(r_2-r_1)t}& \ =\ c_1r_1+(1-c_1)r_2e^{(r_2-r_1)t}+c_1r_1e^{(r_2-r_1)t}-c_1r_1e^{(r_2-r_1)t} \\
		& \ =\ \left(c_1r_1+(1-c_1)r_2\right)e^{(r_2-r_1)t}+c_1r_1\left(1-e^{(r_2-r_1)t}\right).
	\end{align*}
	Therefore, we must have
	\begin{align}
		\label{eqn:first}
		&c_1r_1+(1-c_1)r_2 \leq 0,\\ 
		\label{eqn:second}
		&\left(c_1r_1+(1-c_1)r_2\right)e^{(r_2-r_1)t}+c_1r_1\left(1-e^{(r_2-r_1)t}\right)<0 \text{ for all } t>0,\\ 
		\label{eqn:third}
		&c_2=1-c_1, \text{ and [(i) or (ii)] holds}. 
	\end{align}

	Consider case (i) of \eqref{eqn:third}. Then condition \eqref{eqn:first} holds if and only if $c_1\geq\frac{r_2}{r_2-r_1}$. Given \eqref{eqn:first}, condition \eqref{eqn:second} holds if and only if $(1-c_1)r_2<0$, which is equivalent -- in case (i) -- to $c_1<1$ and $r_2<0$. Thus, in case (i), the first three properties of a discount function are satisfied if and only if $c_2=1-c_1, r_1<r_2<0$ and $\frac{r_2}{r_2-r_1}\leq c_1<1$. Since $r_1<r_2<0$, the fourth property is also satisfied. 
	
	Next, consider case (ii) of \eqref{eqn:third}. The argument is similar to case (i). The condition \eqref{eqn:first} holds if and only if $c_1\leq\frac{r_2}{r_2-r_1}$. Given \eqref{eqn:first}, condition \eqref{eqn:second} holds if and only if $c_1r_1<0$, which is equivalent -- in case (ii) -- to $c_1>0$ and $r_1<0$. Therefore, in case (ii), the first three properties of a discount function are satisfied if and only if $c_2=1-c_1, r_2<r_1<0$ and $0<c_1\leq \frac{r_2}{r_2-r_1}$. Since $r_2<r_1<0$, the fourth property is also satisfied. 
	
	In case (i) of \eqref{eqn:third} let $a=c_2, b=-r_2,$ and $c=r_2-r_1$. It follows that $a, b, c>0$. Since $r_2/r_2-r_1 \leq  c_1 <1 $ it requires that $0<a\leq b/c+1$. We then have $D(t) = ae^{-bt} + (1-a)e^{-(b+c)t}$, where $a, b, c>0$, $a\leq b/c+1$.
	
	Analogously, in case (ii) of \eqref{eqn:third} let $a=c_1, b=-r_1,$ and $c=r_1-r_2$.  We then obtain the same functional from and parameter restrictions as in case (i), that is, $D(t) = ae^{-bt} + (1-a)e^{-(b+c)t}$, where $a, b, c>0$, $a\leq b/c+1$.
\end{proof}
The following proposition demonstrates that these two types of discount function are compatible with the one-switch property. 
\begin{proposition} \label{two}
	Suppose that the preference order $\succcurlyeq$ on ${\mathcal A}$ has a DU representation $(u, D)$, where 
	\begin{itemize}
		\item $D(t) = (1 + c t)e^{-rt}$, where $r\geq c \geq 0$ and $r>0$, or
		\item $D(t) = ae^{-bt} + (1-a)e^{-(b+c)t}$, where $a, b, c>0$, and $a\leq b/c+1$. 		
	\end{itemize}
	Then $\succcurlyeq$ exhibits the one-switch property.
\end{proposition}

\begin{proof}
	The proof adapts Bell's argument \cite[Proposition 2]{bell1988one} to the time preference framework. 
	We need to prove that for any $({\bf x, t})$, $({\bf y, s}) \in {\mathcal A}$ the following function changes sign at most once:
	\[
	\Delta(\sigma)= \sum_{i=1}^nD(t_i+\sigma)u(x_i)-\sum_{j=1}^mD(s_j+\sigma)u(y_j).
	\]
	
	{\em (a) Linear times exponential.} Consider the discount function $D(t)=(1+ct)e^{-rt}$, where $r\geq c\geq 0$ and $r>0$. Then 
	\[
	\Delta(\sigma)  \ =\  \sum_{i=1}^n (1+ct_i+c\sigma)e^{-rt_i}e^{-r\sigma}u(x_i)-\sum_{j=1}^m (1+cs_j+c\sigma)e^{-rs_j}e^{-r\sigma}u(y_j).
	\]
	Rearranging,
	\begin{align*}
		\Delta(\sigma)  \ =\ & e^{-r\sigma} \left(\sum_{i=1}^n e^{-rt_i}u(x_i) - \sum_{j=1}^m e^{-rs_j}u(y_j) \right)\\
		&+e^{-r\sigma}c \left(\sum_{i=1}^n t_ie^{-rt_i}u(x_i) - \sum_{j=1}^ms_j e^{-rs_j}u(y_j) \right)\\
		&+e^{-r\sigma} c \sigma \left(\sum_{i=1}^n e^{-rt_i}u(x_i) - \sum_{j=1}^m e^{-rs_j}u(y_j) \right).
	\end{align*}
	Let 
	\[
	A \ =\ \sum_{i=1}^n e^{-rt_i}u(x_i) - \sum_{j=1}^m e^{-rs_j}u(y_j),
	\]
	and 
	\[
	B \ =\ \sum_{i=1}^n t_ie^{-rt_i}u(x_i) - \sum_{j=1}^m s_j e^{-rs_j}u(y_j).
	\]
	Then 
	\[
	\Delta(\sigma) \ =\ Ae^{-r\sigma}+cBe^{-r\sigma}+cAe^{-r\sigma}\sigma.
	\]
	This expression can be rewritten as follows
	\[
	\Delta(\sigma) \ =\ e^{-r\sigma}\left(A+cB+cA\sigma\right).
	\]
	Since $e^{-r\sigma}>0$, the sign of $\Delta(\sigma)$ equals the sign of $A+cB+cA\sigma$. Since the latter is linear its sign is constant or else changes once at a unique $\sigma$ value.  
	
	{\em (b) Sums of exponentials.} Consider the function $D(t) = ae^{-bt} + (1-a)e^{-(b+c)t}$, where $a, b, c>0$, and $a\leq b/c+1$. Then
	\begin{align*}
		\Delta(\sigma) \ =\ &\sum_{i=1}^n \left(ae^{-bt_i}e^{-b\sigma}+(1-a)e^{-(b+c)t_i}e^{-(b+c)\sigma}\right) u(x_i)\\
		&-\sum_{j=1}^m\left(ae^{-bs_j}e^{-b\sigma}+(1-a)e^{-(b+c)s_j}e^{-(b+c)\sigma}\right)u(y_j).
	\end{align*}
	It can be rearranged so that			
	\begin{align*}
		\Delta(\sigma) \ =\ & ae^{-b\sigma}\left(\sum_{i=1}^n e^{-bt_i}u(x_i)-\sum_{j=1}^m e^{-bs_j}u(y_j)\right)\\
		&+(1-a)e^{-(b+c)\sigma} \left(\sum_{i=1}^n e^{-(b+c)t_i} u(x_i) - \sum_{j=1}^m e^{-(b+c)s_j}u(y_j) \right).
	\end{align*}
	Denote 
	\[
	\tilde{A} \ =\ \sum_{i=1}^n e^{-bt_i}u(x_i)-\sum_{j=1}^m e^{-bs_j}u(y_j),
	\]
	and 
	\[
	\tilde{B} \ =\ \sum_{i=1}^n e^{-(b+c)t_i} u(x_i) - \sum_{j=1}^m e^{-(b+c)s_j}u(y_j).
	\]
	Then 
	\[
	\Delta(\sigma) \ =\ a\tilde{A}e^{-b\sigma}+(1-a)\tilde{B}e^{-(b+c)\sigma}.
	\]
	This expression can be factorized as follows
	\[
	\Delta(\sigma) \ =\ e^{-b\sigma}\left(a\tilde{A}+(1-a)\tilde{B}e^{-c\sigma}\right).
	\]
	Since $e^{-b\sigma}>0$, the sign of $\Delta(\sigma)$ equals the sign of $a\tilde{A}+(1-a)\tilde{B}e^{-c\sigma}$. Therefore, $\Delta(\sigma)$ is either constant or else changes once at unique $\sigma$ value.
\end{proof}

Since Proposition \ref{two} establishes that linear times exponential discount functions and sum of exponentials discount functions satisfy the one-switch property, they must also satisfy the weak one-switch property.

\subsection{The one-switch property for dated outcomes and monotonic impatience}\label{section:OSPandImpatience}

In this section we consider preferences $\succcurlyeq$ on the set ${\mathcal A_1}$ of dated outcomes. When the preferences $\succcurlyeq$ are restricted to ${\mathcal A_1}$, then a DU representation becomes $U(x,t)=D(t)u(x)$ for any $(x, t) \in {\mathcal A_1}$. Necessary and sufficient conditions for this representation are given in \cite{fishburn1982time}. We assume that  $\succcurlyeq$ satisfy Fishburn and Rubinstein's axioms \cite{fishburn1982time} throughout this section. 

\begin{definition}
	We say that $\succcurlyeq$ exhibits the {\em one-switch property for dated outcomes} if $\succcurlyeq$ exhibits the one-switch property on $\mathcal A_1$.%
	\footnote{The related concept of an ordinal one-switch utility function is introduced in \cite{abbas2015ordinal}.}
\end{definition}
Obviously, if $\succcurlyeq$ exhibits the one-switch property on $\mathcal A$, it implies that $\succcurlyeq$ also exhibits the one-switch property for dated outcomes.

Consider the following notions of decreasing and increasing impatience.

\begin{definition}[\cite{prelec2004decreasing}]
	We say that $\succcurlyeq$ exhibits {\em [strictly] decreasing impatience}, if
	for all $(x, t), (y, s) \in {\mathcal A_1}$ such that $0<x<y$ and for all $t<s$:
	if $(x, t) \sim (y, s)$ then for any $\sigma>0$ we have 
	\begin{equation}\label{DIdef}
		(x, t+\sigma) \  [\prec] \preccurlyeq (y, s+\sigma).
	\end{equation}
	We say that $\succcurlyeq$ exhibits
	\begin{itemize}
		\item {\em [strictly] increasing impatience}, if the preference in \eqref{DIdef} is reversed; 
		\item {\em stationarity, or constant impatience}, if the preference in \eqref{DIdef} is replaced by indifference.%
	\end{itemize}
\end{definition}
When preferences have a DU representation, these properties only restrict the discount function. The next proposition follows directly from the definition.
\begin{proposition}\label{ratio}
	Suppose that $\succcurlyeq$ restricted to ${\mathcal A_1}$ has a DU representation. Then $\succcurlyeq$ exhibits [strictly] DI if and only if 
	\begin{equation}\label{eq:rat}
		\frac{D(t)}{D(t+\sigma)} \ [>] \geq \frac{D(s)}{D(s+\sigma)}, \ \text{for all} \ t, s \ \text{such that} \ t<s, \text{ and every }\sigma>0.
	\end{equation}
	Furthermore, $\succcurlyeq$ exhibits
	\begin{itemize}
		\item {[strictly]} II if and only if the inequality in \eqref{eq:rat} is reversed;
		\item constant impatience if and only if the inequality in \eqref{eq:rat} is replaced by the equality.
	\end{itemize}
\end{proposition}

The following proposition provides a complete characterisation.
\begin{proposition}[\cite{prelec2004decreasing}, \cite{anchugina2016aggregating}]\label{prelec} 
	Suppose that $\succcurlyeq$ restricted to ${\mathcal A_1}$ has a DU representation.%
	Then $\succcurlyeq$ exhibits 
	\begin{itemize}
		\item {[strictly]} DI if and only if $D(t)$ is [strictly] log-convex,%
		\footnote{A function $f\colon I \to \mathbb{R}$ is called {\em log-convex} if $f(x)>0$ for all $x\in I$ and $\ln(f)$ is convex; and {\em strictly log-convex} if $f(x)>0$ for all $x\in I$  and $\ln(f)$ is strictly convex. We say that a function $f\colon I \to \mathbb{R}$ is {\em [strictly] log-concave}, if $1/f$ is [strictly] log-convex.} 
		\item {[strictly]} II if and only if $D(t)$ is [strictly] log-concave,
		\item constant impatience if and only if $D(t)=e^{-rt}$ with $r>0$.
	\end{itemize}
\end{proposition}
Proposition \ref{prelec} extends \cite[Corollary 1]{prelec2004decreasing} to increasing impatience and strictly decreasing (and strictly increasing) impatience. The proof is omitted here, since it only requires a minor adjustment of Prelec's original proof.%
\footnote{The proof of the proposition can also be found in the working paper \cite{anchugina2016aggregating}.}  

The following lemma re-expresses the definition of the one-switch property for dated outcomes in terms of a common advancement ($\sigma<0$) and a common delay ($\sigma>0$) applied to a pair of dated outcomes between which the decision-maker is indifferent. 
\begin{lemma}\label{bothways}
	Let $(\hat{x}, \hat{t}), (\hat{y}, \hat{s}) \in {\mathcal A_1}$ such that $0<\hat{x}<\hat{y}$, $0< \hat{t}<\hat{s}$ and $(\hat{x}, \hat{t}) \sim (\hat{y}, \hat{s})$. Then $\succcurlyeq$ exhibits the one-switch property for dated outcomes only if either
	\begin{enumerate}[(i)]
		\item $(\hat{x}, \hat{t}+\sigma) \sim (\hat{y}, \hat{s}+\sigma)$ for all $\sigma\geq -\hat{t}$, or 
		\item $(\hat{x}, \hat{t}+\sigma) \succ (\hat{y}, \hat{s}+\sigma)$ for $-\hat{t}\leq\sigma<0$, and $(\hat{x}, \hat{t}+\sigma) \prec (\hat{y}, \hat{s}+\sigma)$ for $\sigma>0$, or
		\item $(\hat{x}, \hat{t}+\sigma) \prec (\hat{y}, \hat{s}+\sigma)$ for $-\hat{t}\leq\sigma<0$, and $(\hat{x}, \hat{t}+\sigma) \succ (\hat{y}, \hat{s}+\sigma)$ for $\sigma>0$.
	\end{enumerate}
\end{lemma}
\begin{proof}
	Let $(\hat{x}, \hat{t}), (\hat{y}, \hat{s}) \in {\mathcal A_1}$ such that $0<\hat{x}<\hat{y}$, $0< \hat{t}<\hat{s}$ and $(\hat{x}, \hat{t}) \sim (\hat{y}, \hat{s})$. Assume also that $\succcurlyeq$ exhibits the one-switch property for dated outcomes. Therefore, by definition of the one-switch property we have
	\begin{enumerate}[(i$^*$)]
		\item $(\hat{x}, \hat{t}+\sigma) \sim (\hat{y}, \hat{s}+\sigma)$ for all $\sigma>0$, or 
		\item $(\hat{x}, \hat{t}+\sigma) \prec (\hat{y}, \hat{s}+\sigma)$ for all $\sigma>0$, or
		\item $(\hat{x}, \hat{t}+\sigma) \succ (\hat{y}, \hat{s}+\sigma)$ for all $\sigma>0$.
	\end{enumerate}
	We need to analyse the situation when $-\hat{t}\leq \sigma<0$. The proof is by contradiction.
	
	In case (i$^*$), assume that there exists $\mu^*$ such that $0<\mu^*<\hat{t}$ and $(\hat{x}, \hat{t}-\mu^*) \prec (\hat{y}, \hat{s}-\mu^*)$. Let $\hat{\tau}=\hat{t}-\mu^*>0$ and $\hat{\rho}=\hat{s}-\mu^*>0$. Using this notation, we obtain
	\begin{align*}
		(\hat{x}, \hat{\tau}) &\prec (\hat{y}, \hat{\rho}), \\
		(\hat{x}, \hat{\tau}+\mu^*+\sigma) &\sim (\hat{y}, \hat{\rho}+\mu^*+\sigma), \text{ for all } \sigma\geq 0.
	\end{align*}
	
	This contradicts the one-switch property for dated outcomes, so (i) follows. If $(\hat{x}, \hat{t}-\mu^*) \succ (\hat{y}, \hat{s}-\mu^*)$
	the proof is analogous.
	
	In case (ii$^*$), assume that there exists $\mu^*$ such that $0<\mu^*<\hat{t}$ and $(\hat{x}, \hat{t}-\mu^*) \preccurlyeq (\hat{y}, \hat{s}-\mu^*)$. With the same notation as in the previous case $\hat{\tau}=\hat{t}-\mu^*>0$ and $\hat{\rho}=\hat{s}-\mu^*>0$ gives us
	\begin{align*}
		(\hat{x}, \hat{\tau}) &\preccurlyeq (\hat{y}, \hat{\rho}), \\
		(\hat{x}, \hat{\tau}+\mu^*) &\sim (\hat{y}, \hat{\rho}+\mu^*),\\
		(\hat{x}, \hat{\tau}+\mu^*+\sigma) &\prec (\hat{y}, \hat{\rho}+\mu^*+\sigma), \text{ for all } \sigma>0,
	\end{align*}
	which is a contradiction.
	
	In case (iii$^*$) the proof is symmetric to case (ii$^*$).
\end{proof}
The relation between impatience properties and the one-switch property for dated outcomes is established in the following lemma.

\begin{lemma}\label{1sismon}
	Suppose that $\succcurlyeq$ has a DU representation $(u, D)$. Then $\succcurlyeq$ exhibits the one-switch property for dated outcomes if and only if  $\succcurlyeq$ also exhibits either stationarity or strictly DI or strictly II.
\end{lemma}
\begin{proof}
	{\em ``Only If".} Assume that $\succcurlyeq$ exhibits the one-switch property for dated outcomes. 
	
	Consider some $(\hat{x}, \hat{t}), (\hat{y}, \hat{s}) \in {\mathcal A_1}$ such that $0<\hat{x}<\hat{y}$, $0< \hat{t}<\hat{s}$ and $(\hat{x}, \hat{t}) \sim (\hat{y}, \hat{s})$. To see that we can always find such a pair, suppose that $0<\hat{x}<\hat{y}$, $\hat{t}<\hat{s}$ and $(\hat{x}, \hat{t}) \succ (\hat{y}, \hat{s})$. Then it follows by the continuity of $u$ and the fact that $D$ is strictly decreasing that there exists $t' \in (\hat{t}, \hat{s})$ such that $(\hat{x}, t') \sim (\hat{y}, \hat{s})$. Alternatively, suppose that $0<\hat{x}<\hat{y}$, $\hat{t}<\hat{s}$ and $(\hat{x}, \hat{t}) \prec (\hat{y}, \hat{s})$. Then it follows by the continuity of $u$ and the fact that $D$ is strictly decreasing that there exists $x' \in (\hat{x}, \hat{y})$ such that $(x', \hat{t})\sim (\hat{y}, \hat{s})$.
	
	It follows by the one-switch property for dated outcomes and Lemma \ref{bothways} that either
	\begin{enumerate}[{\em {Case} 1.}]
		\item $(\hat{x}, \hat{t}+\sigma) \sim (\hat{y}, \hat{s}+\sigma)$ for all $\sigma\geq -\hat{t}$, or 
		\item $(\hat{x}, \hat{t}+\sigma) \succ (\hat{y}, \hat{s}+\sigma)$ for $-\hat{t}\leq\sigma<0$, and $(\hat{x}, \hat{t}+\sigma) \prec (\hat{y}, \hat{s}+\sigma)$ for $\sigma>0$, or
		\item $(\hat{x}, \hat{t}+\sigma) \prec (\hat{y}, \hat{s}+\sigma)$ for $-\hat{t}\leq\sigma<0$, and $(\hat{x}, \hat{t}+\sigma) \succ (\hat{y}, \hat{s}+\sigma)$ for $\sigma>0$.
	\end{enumerate}
	We will analyse each case separately.
	
	{\em Case 1.} Note that letting $\alpha=\sigma+\hat{t}\geq 0$ and $\hat{\sigma}=\hat{s}-\hat{t}>0$ we have
	\begin{equation*}
		(\hat{x}, \alpha) \sim (\hat{y}, \hat{\sigma}+\alpha) \text{ for all } \alpha \geq 0.
	\end{equation*}
	Using the DU representation it follows that
	\begin{equation}\label{eqn:start}
		\frac{u(\hat{x})}{u(\hat{y})}=\frac{D(\alpha+\hat{\sigma})}{D(\alpha)} \text{ for all } \alpha \geq 0.
	\end{equation}
	Consider some $t'<s'$. Then \eqref{eqn:start} implies
	\begin{equation*}
		\frac{u(\hat{x})}{u(\hat{y})} \ =\ \frac{D(t'+\hat{\sigma})}{D(t')} \ =\ \frac{D(s'+\hat{\sigma})}{D(s')}.
	\end{equation*}
	Rearranging, 
	\begin{equation}\label{eqn:disc}
		\frac{D(s')}{D(t')}=\frac{D(s'+\hat{\sigma})}{D(t'+\hat{\sigma})}.
	\end{equation}
	By continuity we can choose $x'<y'$ such that 
	\begin{equation} \label{eqn:sim}
		\frac{D(s')}{D(t')}=\frac{u(x')}{u(y')}.
	\end{equation}
	Hence, it follows from \eqref{eqn:disc}, \eqref{eqn:sim} that 
	\[
	(x', t') \sim (y', s') \text{ and } (x', t'+\hat{\sigma}) \sim (y', s'+\hat{\sigma}).
	\]
	Then the one-switch property implies that 
	\begin{equation*}
		\frac{D(s')}{D(t')}=\frac{D(s'+\mu)}{D(t'+\mu)} \text{ for all }\mu>0.
	\end{equation*}
	Since $t'<s'$ were arbitrary, it follows by Proposition \ref{ratio} that $\succcurlyeq$ exhibits constant impatience.
	
	{\em Case 2.} Defining $\alpha, \hat{\sigma}$ as for Case 1, we have 
	\begin{align*}
		(\hat{x}, \alpha) &\succ (\hat{y}, \hat{\sigma}+\alpha) \text{ for } 0\leq\alpha<\hat{t}, \text{ and}\\
		(\hat{x}, \alpha) &\prec (\hat{y}, \hat{\sigma}+\alpha) \text{ for } \sigma>\hat{t}. 
	\end{align*}
	Therefore, using the DU representation
	\begin{align*}
		\frac{u(\hat{x})}{u(\hat{y})} &> \frac{D(\hat{\sigma}+\alpha)}{D(\alpha)} \text{ for } 0\leq\alpha<\hat{t}, \text{ and}\\
		\frac{u(\hat{x})}{u(\hat{y})} &< \frac{D(\hat{\sigma}+\alpha)}{D(\alpha)}\text{ for } \alpha>\hat{t}. 
	\end{align*}
	Hence,
	\begin{equation*}
		\frac{D(t'+\hat{\sigma})}{D(t')}<\frac{u(\hat{x})}{u(\hat{y})}<\frac{D(s'+\hat{\sigma})}{D(s')} \text{ for any } t'<\hat{t}<s'.
	\end{equation*}
	Rearranging 
	\begin{equation}\label{eqn:in}
		\frac{D(s')}{D(t')}<\frac{D(s'+\hat{\sigma})}{D(t'+\hat{\sigma})} \text{ for any } t'<\hat{t}<s'.
	\end{equation}
	By continuity we can choose $x'<y'$ such that 
	\begin{equation}\label{eqn:si}
		\frac{D(s')}{D(t')}=\frac{u(x')}{u(y')}.
	\end{equation}
	It follows from \eqref{eqn:in}, \eqref{eqn:si} and the one-switch property that 
	\begin{equation} \label{eqn:keyrat}
		\frac{D(s'-\mu)}{D(t'-\mu)}<\frac{D(s')}{D(t')}<\frac{D(s'+\mu)}{D(t'+\mu)}
	\end{equation}
	for any $\mu>0$ whenever $t'<\hat{t}<s'$.
	
	Consider some $t'<s'$. There are three possible sub-cases:
	\begin{enumerate}[(a)]	
		\item $t'<\hat{t}<s'$, \label{a}
		\item $t'<s'\leq \hat{t}$, and \label{b}
		\item $\hat{t}\leq t'<s'$. \label{c}
	\end{enumerate}
	We will show that, in each of these three sub-cases,		
	\begin{equation} \label{eqn:sub-case}
		\frac{D(s')}{D(t')}<\frac{D(s'+\mu)}{D(t'+\mu)} \text{ for all }\mu>0.
	\end{equation}
	From Proposition \ref{ratio} we may then conclude that the preferences exhibit strict DI.
	
	In sub-case \eqref{a}, it follows directly from \eqref{eqn:keyrat} that \eqref{eqn:sub-case} holds. 
	
	In sub-case \eqref{b} choose $\varepsilon>0$ such that $s'+\varepsilon<\hat{t}<t'+\varepsilon$. Let $t''=t'+\varepsilon$ and $s''=s'+\varepsilon$. It follows from \eqref{eqn:keyrat} that 
	\begin{equation}
		\frac{D(s''-\sigma)}{D(t''-\sigma)}<\frac{D(s'')}{D(t'')}<\frac{D(s''+\sigma)}{D(t''+\sigma)}
	\end{equation}
	for all $\sigma>0$. Let $\sigma= \varepsilon$. Then we have
	\begin{equation}\label{eqn:eq1}
		\frac{D(s')}{D(t')}<\frac{D(s'+\varepsilon)}{D(t'+\varepsilon)}.
	\end{equation}
	By continuity we can choose $x'<y'$ such that 
	\begin{equation}\label{eqn:eq2}
		\frac{D(s')}{D(t')}=\frac{u(x')}{u(y')}.
	\end{equation}
	Therefore, it follows from \eqref{eqn:eq1}, \eqref{eqn:eq2} and the one-switch property that \eqref{eqn:sub-case} holds.
	
	In sub-case \eqref{c} choose $\varepsilon>0$ such that $t'-\varepsilon<\hat{t}<s'-\varepsilon$. Let $t''=t'-\varepsilon$ and $s''=s'-\varepsilon$. Then it follows from \eqref{eqn:keyrat} that
	\begin{equation}
		\frac{D(s''-\sigma)}{D(t''-\sigma)}<\frac{D(s'')}{D(t'')}<\frac{D(s''+\sigma)}{D(t''+\sigma)}
	\end{equation}
	for all $\sigma>0$. Let $\sigma= \varepsilon$. Then we have
	\begin{equation}\label{eqn:c1}
		\frac{D(s'-\varepsilon)}{D(t'-\varepsilon)}<\frac{D(s')}{D(t')}.
	\end{equation}
	Choose $x'<y'$ such that 
	\begin{equation}\label{eqn:c2}
		\frac{D(s')}{D(t')}=\frac{u(x')}{u(y')}.
	\end{equation}
	Therefore, \eqref{eqn:sub-case} follows by \eqref{eqn:c1}, \eqref{eqn:c2} and the one-switch property. 
	
	Therefore, in Case 2 $\succcurlyeq$ exhibits strictly DI.
	
	Finally, in {\em Case 3} we may show that $\succcurlyeq$ exhibits strictly II by a symmetric proof to that for Case 2.

	{\em ``If".}
	Suppose that there are some $x, y, t, s$ with $t\leq s$ and some $\sigma^*\geq 0$ such that $(x, t+\sigma^*)\sim (y, s+\sigma^*)$.
	It suffices to show that either
	\begin{equation}
		(x, t+\sigma)\sim (y, s+\sigma) \text{ for all } \sigma,
	\end{equation} 
	or
	\begin{align}
		(x, t+\sigma) &\succ (y, s+\sigma) \text{ for all } \sigma< \sigma^*, \text{ and }\\
		(x, t+\sigma) &\prec (y, s+\sigma) \text{ for all } \sigma> \sigma^*.
	\end{align}  
	If $t=s$ then $(x, t+\sigma^*)\sim (y, t+\sigma^*)$. It follows by monotonicity that $x=y$. Hence, $\succcurlyeq$ satisfies the one-switch property for dated outcomes.
	
	Assume that $t<s$ and $\succcurlyeq$ exhibits constant impatience. It follows that $(x, t+\sigma^*+\sigma')\sim (y, s+\sigma^*+\sigma')$ for all $\sigma'>0$, or $(x, t+\sigma)\sim (y, s+\sigma)$ for all $\sigma>\sigma^*$. To show that $\succcurlyeq$ satisfies the one-switch property for dated outcomes we need to demonstrate that $(x, t+\sigma)\sim (y, s+\sigma)$ for all $\sigma<\sigma^*$ such that $t+\sigma\geq 0$, or $(x, t+\sigma^*-\sigma')\sim (y, s+\sigma^*-\sigma')$ for all $\sigma'>\sigma^*$ such that $t+\sigma^*-\sigma'\geq 0$. The proof is by contradiction. Suppose that there exists $\sigma''>\sigma^*$ such that, say, $(x, t+\sigma^*-\sigma'')\succ (y, s+\sigma^*-\sigma'')$ with $t+\sigma^*-\sigma''\geq 0$. By continuity and impatience there exist $t'>t$ such that $(x, t'+\sigma^*-\sigma'')\sim (y, s+\sigma^*-\sigma'')$. Then, since $\succcurlyeq$ exhibits constant impatience it follows that $(x, t'+\sigma^*-\sigma''+\gamma)\sim (y, s+\sigma^*-\sigma''+\gamma)$ for all $\gamma>0$. Let $\gamma=\sigma''>0$. Then we obtain $(x, t'+\sigma^*) \sim (y, s+\sigma^*)$. Since $t'>t$ it implies by impatience that $(x, t'+\sigma^*) \prec (x, t+\sigma^*)$. Hence, $(x, t+\sigma^*) \succ (y, s+\sigma^*)$, a contradiction.
	
	Suppose that $t<s$ and $\succcurlyeq$ exhibits strictly DI. It follows that $(x, t+\sigma^*+\alpha) \prec (y, s+\sigma^*+\alpha)$ for all $\alpha>0$, or $(x, t+\sigma) \prec (y, s+\sigma)$ for all $\sigma>\sigma^*$. We now need to show that $(x, t+\sigma) \succ (y, s+\sigma)$ for all $\sigma<\sigma^*$, or $(x, t+\sigma^*-\sigma') \succ (y, s+\sigma^*-\sigma')$ for all $\sigma'>\sigma^*$ such that $t+\sigma^*-\sigma'\geq 0$. The proof is by contradiction. Suppose there exist $\sigma''>\sigma^*$ such that $(x, t+\sigma^*-\sigma'') \preccurlyeq (y, s+\sigma^*-\sigma'')$ and  $t+\sigma^*-\sigma''\geq 0$.
	
	First consider $(x, t+\sigma^*-\sigma'') \sim (y, s+\sigma^*-\sigma'')$ with $\sigma''>\sigma^*$ and $t+\sigma^*-\sigma''\geq 0$. Then, since $\succcurlyeq$ satisfy strictly DI it follows that $(x, t+\sigma^*-\sigma''+\gamma) \prec (y, s+\sigma^*-\sigma''+\gamma)$ for all $\gamma>0$. Let $\gamma=\sigma''>0$. Then we have  $(x, t+\sigma^*) \prec (y, s+\sigma^*)$, which is a contradiction.  
	Secondly, consider $(x, t+\sigma^*-\sigma'') \prec (y, s+\sigma^*-\sigma'')$ with $\sigma''>\sigma^*$ and $t+\sigma^*-\sigma''\geq 0$. It follows by continuity and impatience that there exist $s'>s$ such that $(x, t+\sigma^*-\sigma'') \sim (y, s'+\sigma^*-\sigma'')$. Hence, since $\succcurlyeq$ exhibit strictly DI it implies that $(x, t+\sigma^*-\sigma''+\gamma) \prec (y, s'+\sigma^*-\sigma''+\gamma)$ for all $\gamma>0$ with $\sigma''>\sigma^*$ and $t+\sigma^*-\sigma''\geq 0$. Let $\gamma=\sigma''$. Then we have $(x, t+\sigma^*) \prec (y, s'+\sigma^*)$. Since $s'>s$ it follows by impatience that $(y, s'+\sigma^*) \prec (y, s+\sigma^*)$, therefore, $(x, t+\sigma^*) \prec (y, s+\sigma^*)$, a contradiction.
	
	If we assume that $\succcurlyeq$ exhibits strictly II, the proof is analogous.
\end{proof}
While the assumption of a DU representation is not necessary for the {\em ``If"} part of the proof, it is essential for our proof of the {\em ``Only If"} part. The necessity of a DU representation for the {\em ``Only If"} result of Lemma \ref{1sismon} remains an open question.

It was demonstrated in the working paper \cite{anchugina2016aggregating}, that for differentiable discount functions a strictly increasing time preference rate corresponds to strictly II, while a strictly decreasing time preference rate corresponds to strictly DI.%
	\footnote{The {\em time preference rate}, $r(t)$, is defined as follows: $r(t)\ =\  -\frac{D'(t)}{D(t)}$.}
The proof is along the lines of \cite[Lemma 11]{anchugina2016aggregating}. We use this result to prove the following proposition.

\begin{proposition}
	Suppose that $\succcurlyeq$ has a DU representation $(u, D)$.
	\begin{itemize}
		\item If $D(t) = (1 + c t)e^{-rt}$, where $r\geq c\geq 0$ and $r>0$, then $\succcurlyeq$ exhibits strictly II when $c>0$ and $\succcurlyeq$ exhibits stationarity when $c=0$.
		\item If $D(t) = ae^{-bt} + (1-a)e^{-(b+c)t}$, where $a, b, c>0$, $a\leq b/c+1$, then $\succcurlyeq$ exhibits strictly DI when $a<1$, strictly II when $1<a\leq b/c+1$ and stationarity when $a=1$.
	\end{itemize}
\end{proposition}

\begin{proof}
	{\em (a) Linear times exponential.} Assume first that $c > 0$. Since $D'(t)=e^{-rt}(c-r-crt)$, the time preference rate is:
	\[ 
	-\frac{D'(t)}{D(t)} \ =\ \frac{e^{-rt}(crt-c+r)}{(1+ct)e^{-rt}} \ =\ \frac{r(1+ct)-c}{1+ct} \ =\ r-\frac{c}{1+ct}.
	\]
	The derivative of time preference rate is
	$c^2(1+ct)^{-2}>0$, therefore, linear times exponential discount function exhibits strictly increasing impatience. 
	Otherwise, if $c=0$, then $D(t)=e^{-rt}$ and the preferences $\succcurlyeq$ exhibit stationarity (see, for example, \cite{fishburn1982time}).
	
	{\em (b) Sums of exponentials.} The time preference rate is
	\begin{align*} 
		-\frac{D'(t)}{D(t)} \ =\ &\frac{abe^{-bt}+(1-a)(b+c)e^{-(b+c)t}}{ae^{-bt}+(1-a)e^{-(b+c)t}}=\frac{e^{-bt} \left[ab+(1-a)(b+c)e^{-ct}\right]}{e^{-bt}\left[a+(1-a)e^{-ct}\right]}\\
		\ =\ &\frac{ab+(1-a)(b+c)e^{-ct}}{a+(1-a)e^{-ct}}.
	\end{align*}
	The derivative of time preference rate is
	\begin{gather*} 
		\left[\frac{ab+(1-a)(b+c)e^{-ct}}{a+(1-a)e^{-ct}}\right]' \\
		\ =\ \frac{-c(1-a)(b+c)e^{-ct}\left[a+(1-a)e^{-ct}\right]+\left[ab+(1-a)(b+c)e^{-ct}\right]c(1-a)e^{-ct}}{\left[a+(1-a)e^{-ct}\right]^2}.
	\end{gather*}
	
	The sign of the derivative depends on the sign of the numerator of this expression:
	\[
		Q(t) = -c(1-a)(b+c)e^{-ct}\left[a+(1-a)e^{-ct}\right]+\left[ab+(1-a)(b+c)e^{-ct}\right]c(1-a)e^{-ct}.
	\]
	Simplifying $Q(t)$:
	\[
		Q(t)=c(1-a)e^{-ct}\left[ab+(1-a)(b+c)e^{-ct}-(b+c)(a+(1-a)e^{-ct})\right]=c^2e^{-ct}a(a-1).
	\]
	Recall that $a>0$ and $a \leq b/c+1$. Therefore, $Q(t)=0$ if $a=1$, $Q(t)$ is strictly negative if $0<a<1$ and $Q(t)$ is strictly positive if $1<a\leq b/c+1$ and $a\neq 1$. Hence, the time preference rate is constant if $a=1$, strictly decreasing%
	if $0<a<1$ and strictly increasing if $1<a \leq b/c+1$. This in turn implies that $\succcurlyeq$ exhibit stationarity if $a=1$, strictly DI if $0<a<1$ and strictly II if $1<a\leq b/c+1$.
\end{proof}
A linear times exponential discount function and two sum of exponentials discount functions, with their associated rates of time preference, are illustrated in Figure \ref{graph:Sumex}.

\noindent
\begin{figure}[h]
	\begin{minipage}{\linewidth}
		\makebox[\linewidth]{
			\includegraphics{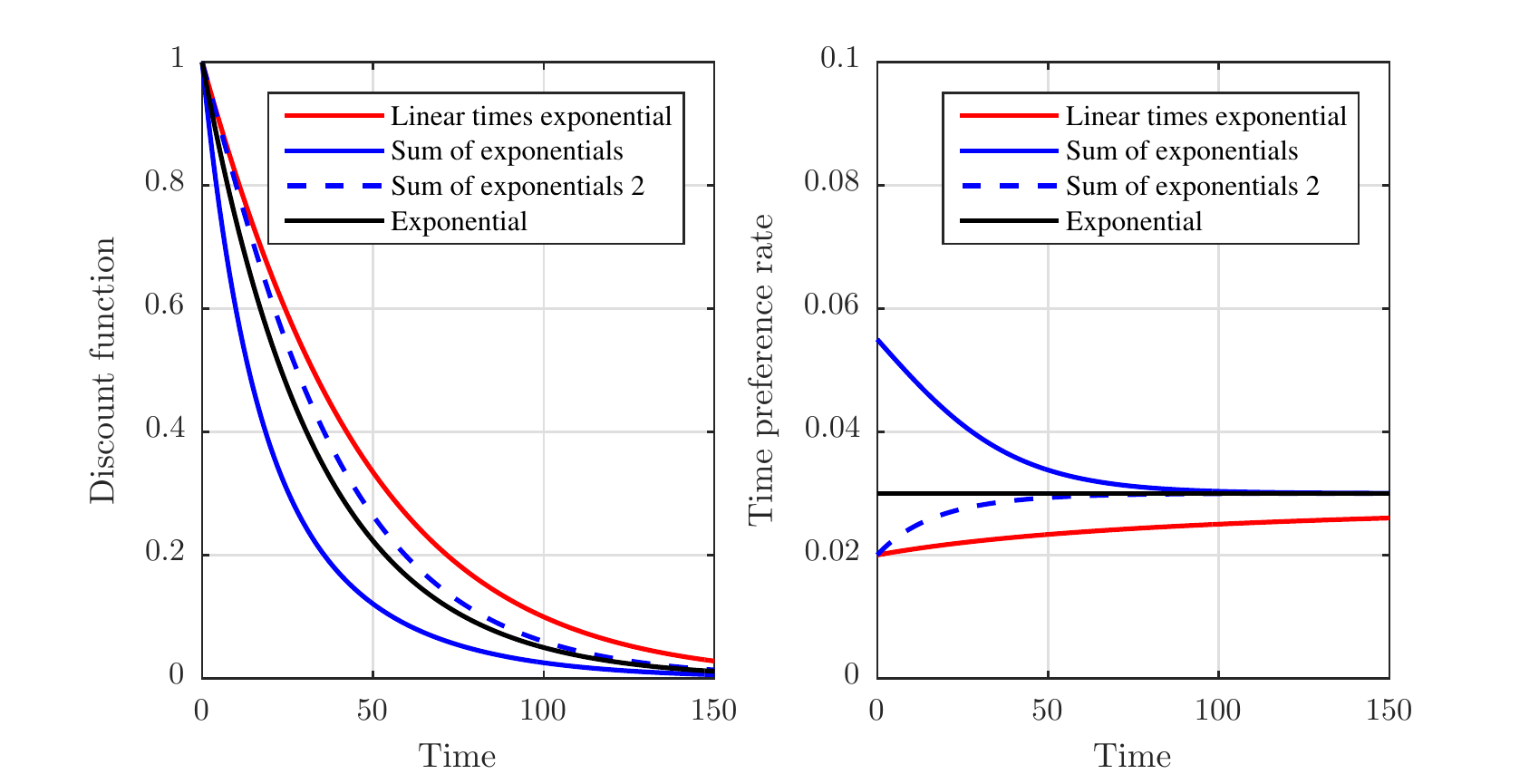}}\
		\captionof{figure}{Linear times exponential discount function $D(t)=(1+0.01t)e^{-0.03t}$, exponential discount function $D(t)=e^{-0.03t}$, 
			sum of exponentials $D(t)= 0.5e^{0.03t} + 0.5e^{0.08t}$, sum of exponentials $D(t)= 1.2e^{0.03t} - 0.2e^{0.08t}$
			and their associated rates of time preference}
		\label{graph:Sumex}
	\end{minipage}\\
\end{figure}

It is worth mentioning that Bell's \cite{bell1988one} definitions of the terms ``decreasing impatience" and ``increasing impatience" are different from the ones used here. Bell's definitions are given below:
\begin{definition}[\cite{bell1988one}]
	Let $\succcurlyeq$ on ${\mathcal A_1}$ have a DU representation with a discount function $D$. Then we say that preferences $\succcurlyeq$ exhibit DI$^*$ [II$^*$] if 
	\[
	D(s+t) > [<] \ D(s)D(t) \text{ for any } s, t>0.
	\] 
\end{definition}
Note that Bell's \cite{bell1988one} DI$^*$ (II$^*$) corresponds to strict log-superadditivity (log-\linebreak subadditivity) of $D$. Obviously, strict log-superadditivity (strict log-subadditivity) is a special case of strict log-convexity (strict log-concavity). Therefore, Bell's definitions of DI$^*$ and II$^*$ are weaker properties than our strictly DI and strictly II. Bell \cite[Proposition 8]{bell1988one} specifies the parameter values for a sum of exponentials discount function such that $\succcurlyeq$ exhibit DI$^*$/II$^*$:
\begin{proposition}[{\cite[Proposition 8]{bell1988one}}] \label{prop:belldiiistar}
	Let $D(t)=ae^{-bt}+(1-a)e^{-(b+c)t}$, where $a, b, c >0$ and $a\leq 1+b/c$. Then it is DI$^*$ if $a<1$ and II$^*$ if $1<a\leq1+b/c$.%
	\footnote{Bell \cite[Proposition 8]{bell1988one} uses a strict inequality $a<1+b/c$ to guarantee that $D'(t)<0$ when $t=0$. However, if $D'(t)=0$ when $t=0$ and $D'(t)<0$ for all $t>0$, then $D(t)$ is strictly decreasing on $[0, \infty)$. Therefore, we allow $a\leq  1+b/c$, since this weak inequality is consistent with the properties of a discount function.}
\end{proposition}

Comparing Proposition \ref{prop:belldiiistar} to Proposition \ref{1sismon}, it is easy to see that strictly II and II$^*$ have the same implications for the parameters for a sum of exponentials discount function (and similarly for strictly DI and DI$^*$). It is also straightforward to observe that the restrictions imposed by the properties of strictly II and II$^*$ on the parameters for a linear times exponential discount function coincide.

\subsection{Representation of preferences with the one-switch \\ property}\label{section:representation}

We first observe that constant impatience is equivalent to the zero-switch property:
\begin{definition}
	We say that $\succcurlyeq$ on $\mathcal{A}$ exhibit the zero-switch property if for every pair $({\bf x, t}), ({\bf y, s}) \in {\mathcal A}$ the ranking of $({\bf x}, {\bf t}+\sigma)$ and $({\bf y}, {\bf s}+\sigma)$ is independent of $\sigma$.
\end{definition}
It follows from this definition that if $\succcurlyeq$ on ${\mathcal A}$ exhibit the zero-switch property, then for any $({\bf x, t}), ({\bf y, s}) \in {\mathcal A}$, if there exists $\hat{\sigma} \geq 0$ such that $({\bf x}, {\bf t}+\hat{\sigma}) \sim ({\bf y}, {\bf s}+\hat{\sigma})$, then $({\bf x}, {\bf t}+\sigma) \sim ({\bf y}, {\bf s}+\sigma)$ for any $\sigma\geq 0$.

The following proposition amends \cite[Proposition 8]{bell1988one}.
\begin{proposition}\label{bellanalogue}
	Let $\succcurlyeq$ on ${\mathcal A}$ have a DU representation $(u, D)$.
	Then $\succcurlyeq$ exhibit the one-switch property only if $D(t)$ has one of the following forms:
	\begin{itemize}
		\item $D(t)=ae^{-bt}+(1-a)e^{-(b+c)t}$, with $a, b, c>0$ and $a\leq b/c+1$, 
		\item $D(t) = (1 + c t)e^{-rt}$, where $r \geq c \geq 0$ and $r>0$.
	\end{itemize}
\end{proposition}

We adapt Bell's \cite{bell1988one, bell1995risk} method of proof for the risk (expected utility) context to our time preference framework. The required adaptation is non-trivial. The main reason is that probabilities sum up to one, whereas utilities of outcomes do not. In the original Bell proof \cite{bell1988one} this property of probabilities was used to obtain a system of two equations in two variables. The conditions under which the solutions of this system exist are well-known. In the time preference setting, we use Lemma \ref{1sismon} to obtain two sequences of dated outcomes that are indifferent at two different points of delay. As a result of which we obtain a homogeneous second order linear difference equation. The solutions of this equation extended to continuous time give us three types of functions. We further show that only linear times exponential and sums of exponentials (with suitable parameter restrictions) satisfy the one-switch property and the properties of a discount function (using Proposition \ref{discfun}). It is also worth mentioning that in the risk setting Bell \cite{bell1988one, bell1995risk} obtains a third order difference equation, rather than the second order difference equation that we obtain in the time preference framework.


\begin{proof} 	
	Fix some $\sigma>0$. Suppose we can find $u(\alpha), u(\beta), u(\gamma)>0$ such that
	\begin{equation} \label{eqn:discrete}
		u(\alpha)+u(\beta)D(2\sigma)=u(\gamma)D(\sigma) \text{ and } u(\alpha)D(\sigma)+u(\beta)D(3\sigma)=u(\gamma)D(2\sigma) 
	\end{equation}
	Then, since $\succcurlyeq$ has a DU representation, it implies that 
	\[
	({\bf x}, {\bf t}) \sim ({\bf y}, {\bf s}) \text{ and } ({\bf x}, {\bf t}+\sigma) \sim ({\bf y}, {\bf s}+\sigma), 
	\]
	where $({\bf x}, {\bf t})=((\alpha,\beta),(0,2\sigma))$, $({\bf y}, {\bf s})=(\gamma, \sigma)$ and $({\bf x}, {\bf t}+\sigma)=((\alpha,\beta),(\sigma,3\sigma))$, $({\bf y}, {\bf s}+\sigma)=(\gamma, 2\sigma)$.
	
	We first show that for any $D(\sigma), D(2\sigma), D(3\sigma)$ we can always find $u(\alpha), u(\beta), u(\gamma)>0$ such that \eqref{eqn:discrete} holds.
	The first equation of \ref{eqn:discrete} implies $u(\alpha)=u(\gamma)D(\sigma)-u(\beta)D(2\sigma)$, which we may substitute into the second equation as follows:
	\[
	\left(u(\gamma)D(\sigma)-u(\beta)D(2\sigma)\right) D(\sigma)+u(\beta)D(3\sigma)=u(\gamma)D(2\sigma).
	\] 
	Simplifying 
	\begin{equation}\label{eqn:simp}
		u(\gamma)(D(\sigma)^2-D(2\sigma))+u(\beta)\left(D(3\sigma)-D(\sigma)D(2\sigma)\right)=0.
	\end{equation}
	Since preferences $\succcurlyeq$ satisfy the one-switch property they will also satisfy the one-switch property for dated outcomes, and hence it follows by Lemma \ref{1sismon} that $\succcurlyeq$ exhibits strictly DI, strictly II or stationarity.
	Assume first that $\succcurlyeq$ exhibits strictly DI. Then, by Proposition \ref{ratio} we have
	\[
	\frac{1}{D(\sigma)} > \frac{D(\sigma)}{D(2\sigma)}>\frac{D(2\sigma)}{D(3\sigma)}>\cdots
	\]
	Therefore, $D(\sigma)^2-D(2\sigma)<0$ and $D(3\sigma)-D(\sigma)D(2\sigma)>0$, and hence, by continuity it is always possible to find $u(\beta), u(\gamma)>0$ such that \eqref{eqn:simp} holds.
	
	Since $\frac{u(\beta)}{u(\gamma)}=\frac{D(2\sigma)-D(\sigma)^2}{D(3\sigma)-D(\sigma) D(2\sigma)}$, we have
	\[
	\frac{u(\beta)}{u(\gamma)}-\frac{D(\sigma)}{D(2\sigma)} \ =\ \frac{D(2\sigma)-D(\sigma)^2}{D(3\sigma)-D(\sigma) D(2\sigma)}-\frac{D(\sigma)}{D(2\sigma)} \ =\ \frac{D(2\sigma)^2-D(\sigma)D(3\sigma)}{D(2\sigma)\left(D(3\sigma)-D(\sigma) D(2\sigma)\right)}<0.
	\]
	Hence, $u(\beta)D(2\sigma) -u(\gamma)D(\sigma)<0$, which implies that 
	\begin{equation*}
		u(\alpha)=u(\gamma)D(\sigma)-u(\beta)D(2\sigma)>0. 
	\end{equation*}	
	If $\succcurlyeq$ exhibits strictly II the argument is analogous with the inequalities reversed. 
	
	In the case when $\succcurlyeq$ exhibits stationarity, we have $D(\sigma)^2-D(2\sigma)=0$ and $D(3\sigma)-D(\sigma)D(2\sigma)=0$, therefore \eqref{eqn:simp} holds for any $u(\beta), u(\gamma)$. Hence, we can choose $u(\beta), u(\gamma)>0$ such that $u(\alpha)=u(\gamma)D(\sigma)-u(\beta)D(2\sigma)>0$.
	
	Therefore, we have found two sequences of dated outcomes $({\bf x}, {\bf t}), ({\bf y}, {\bf s})\in {\mathcal A}$ such that 
	\[
	({\bf x}, {\bf t}+t) \sim ({\bf y}, {\bf s}+t), \text{ where } t =0, \sigma. 
	\]
	Then, since $\succcurlyeq$ satisfies the one-switch property it must be true that 
	\[
	({\bf x}, {\bf t}+t) \sim ({\bf y}, {\bf s}+t) \text{ for any } t\geq 0.
	\]
	In particular, $u(\alpha)D(t)+u(\beta)D(t+2\sigma)=u(\gamma)D(t+\sigma)$.
	
	Since $u(\gamma)\neq 0$, we can write $D(t+2\sigma)+aD(t+\sigma)+bD(t)=0$, where $a=-u(\gamma)/u(\beta)$, $b=u(\alpha)/u(\beta)$. 
	
	For some $\sigma>0$ and some $t_0\geq 0$ let $D_{n}^{(t_0, \sigma)}=D(t_0+n\sigma)$, $n\in {\mathbb N}_0$. We then obtain a homogeneous second order linear difference equation
	\begin{equation}\label{eqn:secodk}
		D_{n+2}^{(t_0, \sigma)}+aD_{n+1}^{(t_0, \sigma)}+bD_n^{(t_0, \sigma)}=0.
	\end{equation} 
	It is well known (see, for example, \cite{mickens1991difference}) that this equation has three types of solutions, which are derived from the characteristic equation: $z^2+az+b=0$. 
	These three solutions are as follows:
	\begin{enumerate}[{\bf {Solution} 1.}]
		\item  If $a^2-4b>0$, then there are two distinct real roots denoted as $z_1$ and $z_2$. In this case the two linearly independent solutions to \eqref{eqn:secodk} are $z_1^n$ and $z_2^n$, where $n\in \mathbb{N}_{0}$. The general solution is $D_n^{(t_0, \sigma)}=c_1z_1^n+c_2z_2^n$, where $c_1, c_2 = const$.
		\item  If $a^2-4b=0$, then the roots coincide so $z_1=z_2=z_0$. In this case the two linearly independent solutions to 
		\eqref{eqn:secodk} are $z_0^n$ and $nz_0^n$. The general solution is $D_n^{(t_0, \sigma)}=\left(c_1+c_2n\right)z_0^n$, where $c_1, c_2 = const$.
		\item  If $a^2-4b<0$, the roots are complex 
		\[
		z_{\pm}\ =\ x\pm iy\ =\ r\left(\cos{\theta} \pm i\sin{\theta}\right)\ =\ re^{\pm i\theta},
		\] 
		where $y>0$, $r=\sqrt{x^2+y^2}$, $\cos \theta =x/r$ and $\sin \theta=y/r$ with $\theta \in (0, \pi)$ (since $y>0$). Then the two linearly independent solutions to \eqref{eqn:secodk} are $Re z_{+}^n=r^n \cos{n \theta}$ and $Im z_{+}^n=r^n \sin{n\theta}$. The general solution is 
		\begin{equation*}\label{oscillate}
			D_n^{(t_0, \sigma)}=r^n\left[c_1\cos n\theta+c_2\sin n \theta\right],
		\end{equation*}
		where $c_1, c_2=const.$
	\end{enumerate}
	
	Note that by letting $C=\sqrt{c_1^2+c_2^2}$, $\cos \omega =\frac{c_1}{C}$, $\sin \omega =\frac{c_2}{C}$ and $\omega=\tan^{-1}(\frac{c_2}{c_1})$, Solution 3 can be rewritten as follows:
	\begin{equation*}
		D_n^{(t_0, \sigma)}=r^n\left[c_1\cos n\theta+c_2\sin n \theta\right]=Cr^n\left[\cos \omega \cos n\theta+\sin \omega\sin n \theta\right]=Cr^n\cos(n\theta-\omega).
	\end{equation*}
	Recall that $\theta \in (0, \pi)$. 
	Therefore, Solution 3 can be excluded because it implies multiple changes of sign (it does not satisfy monotonicity).

	
	Note that equation \eqref{eqn:secodk} holds for arbitrary $\sigma>0$ and arbitrary $t_{0}\geq 0$, though the roots $z_{1}$ and $z_{2}$ and the constants $c_{1}$ and $c_{2}$ may depend, in a continuous fashion, on $t_{0}$ and $\sigma$.  Bell \cite{bell1995risk, bell1988one} argues that D must therefore satisfy the corresponding limit of one of these solutions, as $\sigma\rightarrow 0$.%
	\footnote{%
		Romanian mathematician Radó \cite{rado1962characterization} proved a more general result which was recently extended to multidimensional case in \cite{almira2016on}. Radó \cite[Theorem 2]{rado1962characterization} proves that the set of continuous functions $f:\mathbb{R} \to \mathbb{R}$, which satisfy the equation $a_0(\sigma)f(t)+a_1(\sigma)f(t+\sigma)+\ldots+a_n(\sigma)f(t+n\sigma)=0$ with continuous functions $a_i(\sigma)\colon (0, W) \to \mathbb{R}$, where $W>0$ and $i=0, \ldots, n$ with $a_n(\sigma)\neq 0$, coincides with the set of functions $f\colon \mathbb{R} \to \mathbb{R}$ which satisfy a linear differential equation $A_0f+A_1f'+\ldots+A_nf^{(n)}=0$ for some real coefficients $A_{0},\ldots, A_{n}$. When $n=2$ it is well-known that the solutions to such a differential equation coincide with the limits of our three solution types. 
	}
	The solutions of \eqref{eqn:secodk} converge, respectively to
	
	\begin{enumerate}[{\bf {Solution} 1.}]	
		\item (Sum of exponentials): $D(t) = c_1e^{r_1t} + c_2e^{r_2t}$, where $r_1\neq r_2$,
		\item (Linear times exponential): $D(t) = (c_1 + c_2 t)e^{r_0t}$.
	\end{enumerate}
	By Proposition \ref{discfun} it follows that:
	\begin{enumerate}[{\bf {Solution} 1.}]
		\item $D(t) = ae^{-bt} + (1-a)e^{-(b+c)t}$, where $a, b, c>0$, and $a\leq b/c+1$;
		\item $D(t) = (1 + c t)e^{-rt}$, where $r \geq c \geq 0$ and $r>0$. Note that this solution includes the exponential discount function as a special case. That is, if $c=0$, then $D(t) = e^{-rt}$, where $r > 0$.
	\end{enumerate}	
\end{proof}

The following corollary summarizes Proposition \ref{two} and Proposition \ref{bellanalogue}.
\begin{corollary}
	Let $\succcurlyeq$ on ${\mathcal A}$ have a DU representation $(u, D)$. Then $\succcurlyeq$ exhibits the one-switch property if and only if $D(t)$ has one of the following forms:
	\begin{itemize}
		\item $D(t)=ae^{-bt}+(1-a)e^{-(b+c)t}$, with $a, b, c>0$ and $a\leq b/c+1$,
		\item $D(t) = (1 + c t)e^{-rt}$, where $r \geq c \geq 0$ and $r>0$.
	\end{itemize}	
\end{corollary}

\section{The weak one-switch property} \label{section:w1s}

\subsection{The weak one-switch property and mixtures of sequences of dated outcomes}

The one-switch property is equivalent to the weak one-switch property in the risk setting when preferences over lotteries have an expected utility representation \cite{bell1988one}. This equivalence follows from \cite[Lemma 3]{bell1988one}, where mixture linearity and other properties of expected utility are used to show that if two lotteries are indifferent at two wealth levels then they should be indifferent for any wealth level. In the intertemporal framework of this paper a direct adaptation of the proof of this equivalence is not possible, even if we assume that preferences have a DU representation, because a DU representation is, in general, not mixture linear. However, we can adapt the Anscombe and Aumann (AA) setting \cite{anscombe1963definition} to preferences over streams of consumption lotteries. Working in this environment it is possible to establish the equivalence of the weak one-switch property and the one-switch property for time preferences with a discounted {\em expected} utility representation. We first have to adapt the AA framework to continuous time for this purpose.

Assume that $X$ is a {\bf mixture set} (\cite{fishburn1982foundations}); that is, for every $x, y \in X$ and every $\lambda, \mu \in [0, 1]$, 
there exists $x \lambda y \in X$ satisfying: 
\begin{itemize} \renewcommand*\labelitemi{\textbullet}
	\item $x1y=x$, 
	\item $x \lambda y = y (1-\lambda) x$,
	\item$(x \mu y) \lambda y = x (\lambda \mu) y$.
\end{itemize}

We assume that $X$ contains a ``neutral'' outcome, denoted by $0$. We can think of $X$ as a set of lotteries with monetary outcomes, and $0$ corresponds to the lottery which pays $0$ with certainty.


We next introduce a mixture operation for sequences of dated outcomes, analogous to the AA mixing operation. To do so, we recall that the neutral outcome obtains at any date not specified in the sequence. 
First, define 
\begin{equation}
	({\bf x, t})\lambda ({\bf y, t})=({\bf x} \lambda {\bf y} , {\bf t}) \text{ for any }({\bf x}, {\bf t}), ({\bf y}, {\bf t}) \in {\mathcal A} \text{ and all } \lambda \in [0, 1],
\end{equation}
where ${\bf x} \lambda {\bf y}$ is defined the usual way (see \cite{fishburn1982foundations}).

Let $({\bf x, t}) \in {\mathcal A}$ with ${\bf t}=(t_1, t_2, \ldots, t_n)$ and let ${\bf s}=(s_1, s_2, \ldots, s_m)$. Define ${\bf t\vert s} = {\bf l} = (l_1, l_2, \ldots, l_k),$ where $l_1<l_2<\ldots<l_k$ and $\{l_1, l_2, \ldots, l_k\}=\{t_1, t_2, \ldots, t_n\} \cup \{s_1, s_2, \ldots, s_m\}$. An example of concatenation procedure for time vectors ${\bf t}=(t_1, t_2, t_3)$ and ${\bf s}=(s_1, s_2, s_3, s_4)$ is given in Figure \ref{fig:concatenation}. 

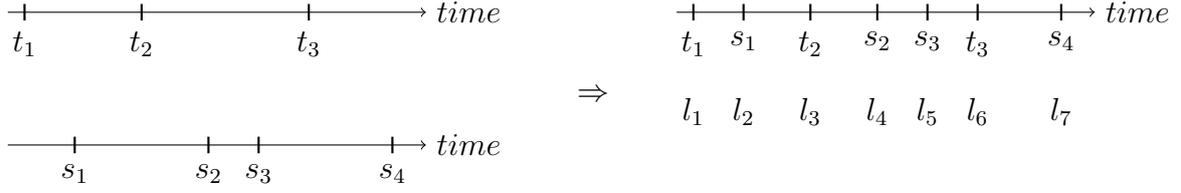
\begin{figure}[!htb] 
	\begin{minipage}{\linewidth}
		\centering
		\begin{tikzpicture}[
		scale=2.2,		axis/.style={very thick, ->, >=stealth', line join=miter},
		every node/.style={color=black},
		dot/.style={circle,fill=black,minimum size=4pt,inner sep=0pt,
			outer sep=-1pt},
		]
		\draw[->] (0, 0)--(2.5,0) node(xline)[right] {$time$}; 
		\draw [thick] (0.4,-.05) node[below]{$s_1$} -- (0.4,0.05);
		\draw [thick] (1.2,-.05) node[below]{$s_2$} -- (1.2,0.05);
		\draw [thick] (1.5,-.05) node[below]{$s_3$} -- (1.5,0.05);
		\draw [thick] (2.3,-.05) node[below]{$s_4$} -- (2.3,0.05);
		\node[] at (3.5,0.3) {$\Rightarrow$};
		\begin{scope}[yshift=0.8cm]
		\draw[->] (0, 0)--(2.5,0) node(xline)[right] {$time$}; 
		\draw [thick] (0.1,-.05) node[below]{$t_1$} -- (0.1,0.05);
		\draw [thick] (0.8,-.05) node[below]{$t_2$} -- (0.8,0.05);
		\draw [thick] (1.8,-.05) node[below]{$t_3$} -- (1.8,0.05);
		\end{scope}
		\begin{scope}[yshift=0.8cm, xshift=4cm]
		\draw[->] (0, 0)--(2.5,0) node(xline)[right] {$time$}; 
		\draw [thick] (0.1,-.05) node[below]{$t_1$} -- (0.1,0.05);
		\draw [thick] (0.4,-.05) node[below]{$s_1$} -- (0.4,0.05);
		\draw [thick] (0.8,-.05) node[below]{$t_2$} -- (0.8,0.05);
		\draw [thick] (1.2,-.05) node[below]{$s_2$} -- (1.2,0.05);
		\draw [thick] (1.5,-.05) node[below]{$s_3$} -- (1.5,0.05);
		\draw [thick] (1.8,-.05) node[below]{$t_3$} -- (1.8,0.05);
		\draw [thick] (2.3,-.05) node[below]{$s_4$} -- (2.3,0.05);
		\node[] at (0.1,-0.6) {$l_1$};
		\node[] at (0.4,-0.6) {$l_2$};
		\node[] at (0.8,-0.6) {$l_3$};
		\node[] at (1.2,-0.6) {$l_4$};
		\node[] at (1.5,-0.6) {$l_5$};
		\node[] at (1.8,-0.6) {$l_6$};
		\node[] at (2.3,-0.6) {$l_7$};
		\end{scope}
		\end{tikzpicture}
	\end{minipage}\\
	\caption{Concatenation of ${\bf t}=(t_1, t_2, t_3)$ and ${\bf s}=(s_1, s_2, s_3, s_4)$}
	\label{fig:concatenation}
\end{figure}

For any $({\bf x, t}), ({\bf y, s}) \in {\mathcal A}$ and any $\lambda \in [0, 1]$ define the mixture operation as follows:
\begin{equation}\label{eqn:mixdef}
	({\bf x, t})\lambda ({\bf y, s})=({\bf z, t\vert s})\lambda ({\bf z', s\vert t}),
\end{equation}
where ${\bf z}$ is defined so that	if $l_j=t_i$, then $z_j=x_i$, otherwise $z_j=0$,
and ${\bf z'}$ is defined so that if $l_j=s_i$, then $z'_j=y_i$, otherwise $z'_j=0$. Note that $({\bf x},{\bf t})$ and $({\bf z, t\vert s})$ are identical sequences of dated outcomes, and likewise $({\bf y},{\bf s})$ and $({\bf z'},{\bf s|t})$ are identical sequences.
Figure \ref{fig:transf} illustrates the transformation of the sequence $({\bf x, t})=((x_1, x_2, x_3), (t_1, t_2, t_3))$ to the sequence $({\bf z}, {\bf t \lvert s})$ and the sequence $({\bf y, s})=((y_1, y_2, y_3, y_4), (s_1, s_2, s_3, s_4))$ to the sequence $({\bf z'}, {\bf s \lvert t})$. Note that ${\bf t} \vert {\bf s}={\bf s} \vert {\bf t}$.

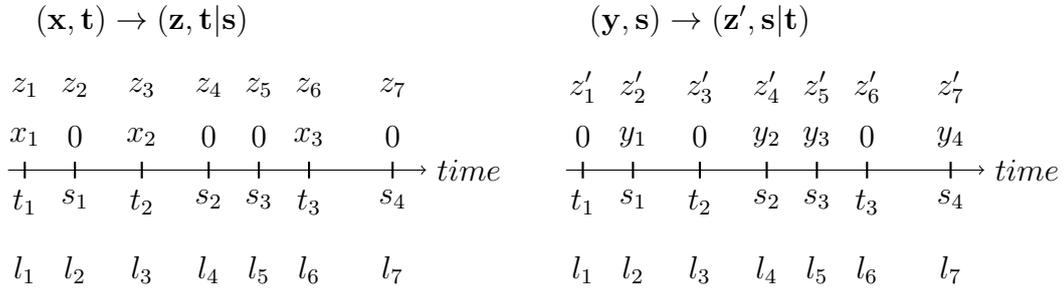
\begin{figure}[!htb]
	\centering
	\begin{minipage}{.4\linewidth}
		\centering
		\begin{tikzpicture}[
		scale=2.2,		axis/.style={very thick, ->, >=stealth', line join=miter},
		every node/.style={color=black},
		dot/.style={circle,fill=black,minimum size=4pt,inner sep=0pt,
			outer sep=-1pt},
		]
		\node[] at (0.8,0.9) {$({\bf x},{\bf t}) \to ({\bf z}, {\bf t \lvert s})$};
		\draw[->] (0, 0)--(2.5,0) node(xline)[right] {$time$}; 
		
		\draw [thick] (0.1,-.05) node[below]{$t_1$} -- (0.1,0.05);
		\draw [thick] (0.8,-.05) node[below]{$t_2$} -- (0.8,0.05);
		\draw [thick] (1.8,-.05) node[below]{$t_3$} -- (1.8,0.05);
		
		\node[] at (0.1,0.2) {$x_1$};
		\node[] at (0.8,0.2) {$x_2$};
		\node[] at (1.8,0.2) {$x_3$};
		
		\draw [thick] (0.4,-.05) node[below]{$s_1$} -- (0.4,0.05);
		\draw [thick] (1.2,-.05) node[below]{$s_2$} -- (1.2,0.05);
		\draw [thick] (1.5,-.05) node[below]{$s_3$} -- (1.5,0.05);
		\draw [thick] (2.3,-.05) node[below]{$s_4$} -- (2.3,0.05);
		
		\node[] at (0.4,0.2) {$0$};
		\node[] at (1.2,0.2) {$0$};
		\node[] at (1.5,0.2) {$0$};
		\node[] at (2.3,0.2) {$0$};
		
		\node[] at (0.1,-0.6) {$l_1$};
		\node[] at (0.4,-0.6) {$l_2$};
		\node[] at (0.8,-0.6) {$l_3$};
		\node[] at (1.2,-0.6) {$l_4$};
		\node[] at (1.5,-0.6) {$l_5$};
		\node[] at (1.8,-0.6) {$l_6$};
		\node[] at (2.3,-0.6) {$l_7$};
		
		\node[] at (0.1,0.5) {$z_1$};
		\node[] at (0.4,0.5) {$z_2$};
		\node[] at (0.8,0.5) {$z_3$};
		\node[] at (1.2,0.5) {$z_4$};
		\node[] at (1.5,0.5) {$z_5$};
		\node[] at (1.8,0.5) {$z_6$};
		\node[] at (2.3,0.5) {$z_7$};
		\end{tikzpicture}
	\end{minipage}
	\qquad 
	\begin{minipage}{.4\linewidth}
		\centering
		\begin{tikzpicture}[
		scale=2.2,		axis/.style={very thick, ->, >=stealth', line join=miter},
		every node/.style={color=black},
		dot/.style={circle,fill=black,minimum size=4pt,inner sep=0pt,
			outer sep=-1pt},
		]
		
		\node[] at (0.8,0.9) {$({\bf y},{\bf s}) \to ({\bf z'}, {\bf s \lvert t})$};
		\draw[->] (0, 0)--(2.5,0) node(xline)[right] {$time$}; 
		\draw [thick] (0.4,-.05) node[below]{$s_1$} -- (0.4,0.05);
		\draw [thick] (1.2,-.05) node[below]{$s_2$} -- (1.2,0.05);
		\draw [thick] (1.5,-.05) node[below]{$s_3$} -- (1.5,0.05);
		\draw [thick] (2.3,-.05) node[below]{$s_4$} -- (2.3,0.05);
		
		\node[] at (0.4,0.2) {$y_1$};
		\node[] at (1.2,0.2) {$y_2$};
		\node[] at (1.5,0.2) {$y_3$};
		\node[] at (2.3,0.2) {$y_4$};
		
		\draw [thick] (0.1,-.05) node[below]{$t_1$} -- (0.1,0.05);
		\draw [thick] (0.8,-.05) node[below]{$t_2$} -- (0.8,0.05);
		\draw [thick] (1.8,-.05) node[below]{$t_3$} -- (1.8,0.05);
		
		\node[] at (0.1,0.2) {$0$};
		\node[] at (0.8,0.2) {$0$};
		\node[] at (1.8,0.2) {$0$};
		
		\node[] at (0.1,-0.6) {$l_1$};
		\node[] at (0.4,-0.6) {$l_2$};
		\node[] at (0.8,-0.6) {$l_3$};
		\node[] at (1.2,-0.6) {$l_4$};
		\node[] at (1.5,-0.6) {$l_5$};
		\node[] at (1.8,-0.6) {$l_6$};
		\node[] at (2.3,-0.6) {$l_7$};
		
		\node[] at (0.1,0.5) {$z'_1$};
		\node[] at (0.4,0.5) {$z'_2$};
		\node[] at (0.8,0.5) {$z'_3$};
		\node[] at (1.2,0.5) {$z'_4$};
		\node[] at (1.5,0.5) {$z'_5$};
		\node[] at (1.8,0.5) {$z'_6$};
		\node[] at (2.3,0.5) {$z'_7$};
		\end{tikzpicture}
	\end{minipage}
	\caption{Transformation of $({\bf x, t})=((x_1, x_2, x_3), (t_1, t_2, t_3))$ to $({\bf z}, {\bf t \lvert s})$ and $({\bf y, s})=((y_1, y_2, y_3, y_4), (s_1, s_2, s_3, s_4))$ to $({\bf z'}, {\bf s \lvert t})$}
	\label{fig:transf}
\end{figure}
It is not hard to see that ${\mathcal A}$ is a mixture set. 

The utility function $u \colon X \to \mathbb{R}$ is called {\em mixture linear} if for every $x, y \in X$ we have $u(x\lambda y)=\lambda u(x)+(1-\lambda) u(y)$.

We say that preferences in this environment have {\em a DEU (discounted expected utility) representation} if they have a DU representation $(u,D)$ in which $u$ is mixture linear on $X$.  We next show that the induced utility $U$ on sequences of dated outcomes is mixture linear on $\cal{A}$. It follows from \eqref{eqn:mixdef} that 
\begin{equation}
	U(({\bf x}, {\bf t})\lambda({\bf y}, {\bf s})) \ =\  U(({\bf z, t\vert s})\lambda ({\bf z', s\vert t})) \ =\ U(({\bf z, l})\lambda ({\bf z', l})) \ =\ U({\bf z}\lambda {\bf z'}, {\bf l}),
\end{equation}
where ${\bf l}= {\bf t\vert s}={\bf s\vert t}$. Since $u$ is mixture linear it follows that 
\begin{align*}
	U({\bf z}\lambda {\bf z'}, {\bf l})& \ =\  \lambda U({\bf z, l})+(1-\lambda)U({\bf z', l}) \ =\ \lambda U({\bf z, t\vert s})+(1-\lambda)U({\bf z', s\vert t})\\
	& \ =\ \lambda U({\bf x, t})+(1-\lambda)U({\bf y, s}).
\end{align*}
Hence, $U$ is mixture linear on $\cal{A}$.

It is worth mentioning that the problem of finding axiomatic foundations for a DEU representation remains an open question. It is natural to assume that Fishburn and Rubinstein' s axioms \cite{fishburn1982time} should be satisfied for $\succcurlyeq$ restricted to ${\mathcal A_1}$ with $X$ being a mixture set.%
\footnote{For example, Bleichrodt et al. \cite{bleichrodt2009non} assume the existence of DU representation for $\succcurlyeq$ on ${\mathcal A_1}$, where $X$ is not necessarily restricted to reals. Similarly, Rohde \cite{rohde2010hyperbolic} applies the DU representation to $\succcurlyeq$ on ${\mathcal A_1}$ requiring only that $X$ is a connected topological space which contains a ``neutral'' outcome.} 
However, it is beyond the scope of the present paper to address this issue.

The following lemma is a part of \cite[Lemma 3]{bell1988one}. The proof is given for completeness.
\begin{lemma}[cf., \cite{bell1988one}]\label{Step1}
	Let preference $\succcurlyeq$ on $\mathcal{A}$ have a DU representation.
	If $\succcurlyeq$ on ${\mathcal A}$ exhibit the weak one-switch property, then for any $({\bf x, t}), ({\bf y, s}) \in {\mathcal A}$ if there exist $\sigma_1, \sigma_2$ such that $\sigma_1<\sigma_2$ and $({\bf x}, {\bf t}+\sigma_1) \sim ({\bf y}, {\bf s}+\sigma_1)$, and 
	$({\bf x}, {\bf t}+\sigma_2) \sim ({\bf y}, {\bf s}+\sigma_2)$, then $({\bf x}, {\bf t}+\sigma) \sim ({\bf y}, {\bf s}+\sigma)$ for any $\sigma \in (\sigma_1, \sigma_2)$.
\end{lemma}
\begin{proof}
	Suppose $({\bf x, t}), ({\bf y, s}) \in {\mathcal A}$ and $\sigma_1, \sigma_2$ are such that $\sigma_1<\sigma_2$ with
	\begin{align*}
		({\bf x}, {\bf t}+\sigma_1) &\sim ({\bf y}, {\bf s}+\sigma_1), \\
		({\bf x}, {\bf t}+\sigma_2) &\sim ({\bf y}, {\bf s}+\sigma_2). 
	\end{align*}
	We need to show that $({\bf x}, {\bf t}+\sigma) \sim ({\bf y}, {\bf s}+\sigma)$ for any $\sigma \in (\sigma_1, \sigma_2)$. The proof is by contradiction. Assume that we can have $\hat{\sigma}\in (\sigma_1, \sigma_2)$ such that $({\bf x}, {\bf t}+\hat{\sigma}) \succ ({\bf y}, {\bf s}+\hat{\sigma})$. Consider ${\bf y'}={\bf y}+\varepsilon$, where $\varepsilon>0$ is sufficiently small so that $({\bf x}, {\bf t}+\hat{\sigma}) \succ ({\bf y'}, {\bf s}+\hat{\sigma})$ (by continuity, such $\varepsilon$ can always be found). However,  $({\bf x}, {\bf t}+\sigma_1) \prec ({\bf y'}, {\bf s}+\sigma_1)$ and $({\bf x}, {\bf t}+\sigma_2) \prec ({\bf y'}, {\bf s}+\sigma_2)$, which implies a double switch. We obtained the desired contradiction. The case $({\bf x}, {\bf t}+\hat{\sigma})\prec ({\bf y}, {\bf s}+\hat{\sigma})$ is similar. Therefore, $({\bf x}, {\bf t}+\sigma) \sim ({\bf y}, {\bf s}+\sigma) \text{ for any }\sigma \in (\sigma_1, \sigma_2)$.
\end{proof}
Lemma \ref{Step1} implies the following corollary:
\begin{corollary}\label{closedinterval}
	The weak one-switch property implies that $\{ \sigma \ \lvert \ \Delta(\sigma)=0\}$ is a closed (possibly empty) interval.	
\end{corollary}
In the next proposition the mixtures of sequences of dated outcomes will be used. It adapts \cite[Lemma 3]{bell1988one} to the time preference setting. Proposition \ref{alevels} is proved by contradiction. We first need to find two sequences of dated outcomes such that their DEU difference changes its sign as the delay increases. Using Corollary \ref{closedinterval} we then consider two cases depending on whether the just obtained DEU difference equals zero at a unique point or on the interval.
A double switch (contradiction to the weak one-switch property) can be obtained in each case by introducing suitable mixtures of sequences of dated outcomes. The illustrations to support the proof are given in the Appendix. 
\begin{proposition}[{cf., \cite[Lemma 3]{bell1988one}}]\label{alevels} 
	Let preference $\succcurlyeq$ on $\mathcal{A}$ have a DEU representation.
	If $\succcurlyeq$ exhibits the weak one-switch property, then for any $({\bf x, t}), ({\bf y, s}) \in {\mathcal A}$ if $({\bf x}, {\bf t}+\sigma_1) \sim ({\bf y}, {\bf s}+\sigma_1)$ and 
	$({\bf x}, {\bf t}+\sigma_2) \sim ({\bf y}, {\bf s}+\sigma_2)$ for some $\sigma_1\geq 0$ and some $\sigma_2>\sigma_1$, then $({\bf x}, {\bf t}+\sigma) \sim ({\bf y}, {\bf s}+\sigma)$ for any $\sigma\geq0$.
\end{proposition}
\begin{proof}
	Suppose $0 \leq \sigma_1<\sigma_2$ and
	\begin{align*}
		({\bf x}, {\bf t}+\sigma_1) &\sim ({\bf y}, {\bf s}+\sigma_1), \\
		({\bf x}, {\bf t}+\sigma_2) &\sim ({\bf y}, {\bf s}+\sigma_2). 
	\end{align*}
	By Lemma \ref{Step1} it follows that 	
	\begin{equation}\label{c3}
		({\bf x}, {\bf t}+\sigma) \sim ({\bf y}, {\bf s}+\sigma) \text{ for any }\sigma \in (\sigma_1, \sigma_2).
	\end{equation}
	The weak one-switch property implies weak preference in one direction above $\sigma_2$ and in the other direction below $\sigma_1$. We assume that $({\bf x}, {\bf t}+\sigma)$ is weakly preferred to $({\bf y}, {\bf s}+\sigma)$ above  $\sigma_2$. The other case can be treated similarly.
	
	Suppose that \eqref{c3} is satisfied but 
	\begin{equation}\label{an}
		({\bf x}, {\bf t}+\sigma') \succ ({\bf y}, {\bf s}+\sigma') \text{ for some }\sigma'> \sigma_2. 
	\end{equation}
	We will prove that a contradiction follows.
	
	Let $\Delta_1(\sigma)=U({\bf x}, {\bf t}+\sigma)-U({\bf y}, {\bf s}+\sigma)$. Therefore,
	\[
	\Delta_1(\sigma)=0, \text{ if } \sigma \in [\sigma_1, \sigma_2], \text{ and } \Delta_1(\sigma')>0.
	\]
	Given the DEU representation we can always find $a, b\in X$ and $q>p$ such that 
	\[
	\Delta_2(\sigma)=U(a, q+\sigma)-U(b, p+\sigma)<0 \text{ for all } \sigma.
	\] 
	Let $({\bf z_1, t_1})=(a, q) \lambda ({\bf x}, {\bf t})$, where $\lambda \in (0, 1)$. Analogously, define $({\bf z_2, t_2}) = (b, p)\lambda({\bf y}, {\bf s})$. Consider the function
	\begin{align*}
		\Delta_3(\sigma)& \ =\ U({\bf z_1, t_1}+\sigma)-U({\bf z_2, t_2}+\sigma)\\
		& \ =\ \lambda \left(U(a, q+\sigma)-U(b, p+\sigma)\right)+(1-\lambda) \left(U({\bf x}, {\bf t}+\sigma)-U({\bf y}, {\bf s}+\sigma)\right)\\
		& \ =\ \lambda \Delta_2(\sigma)+(1-\lambda)\Delta_1(\sigma).
	\end{align*}
	Then for $\lambda$ sufficiently small we have (see Figure \ref{Delta3} in the Appendix%
	\footnote{Note that for all figures in Appendix only sign is relevant in the vertical dimension.}%
	):
	\[
	\Delta_3(\sigma_1)<0, \Delta_3(\sigma_2)<0, \text{ and } \Delta_3(\sigma')>0.
	\]		
	Therefore, by continuity, there must exist at least one $\sigma^* \in (\sigma_2, \sigma')$ such that $\Delta_3(\sigma^*)=0$. By Corollary \ref{closedinterval} there are two possible cases: either $\sigma^*$ is unique, or $\Delta_3(\sigma)=0$ for $\sigma \in [\sigma^*_1, \sigma^*_2] \subset (\sigma_2, \sigma')$ (see Figure \ref{Two_options} in the Appendix). 
	
	{\em Case 1}. Assume first that $\sigma^*$ is unique; i.e., (using the one-switch property) $\Delta_3(\sigma)<0$ if $\sigma<\sigma^*$ and $\Delta_3(\sigma)>0$ if $\sigma>\sigma^*$. 
	Consider the reflection of $\Delta_3(\sigma)$ about the $\sigma$-axis; i.e., $\hat{\Delta}_3(\sigma)=-\Delta_3(\sigma)$. Then $\hat{\Delta}_3(\sigma)>0$ if $\sigma<\sigma^*$ and $\hat{\Delta}_3(\sigma)<0$ if $\sigma>\sigma^*$. Next, choose $\hat{\sigma}\in (\sigma_1, \sigma_2)$ and shift the function $\hat{\Delta}_3(\sigma)$ to the left as follows
	\begin{align*}
		\tilde{\Delta}_3(\sigma)& \ =\ \hat{\Delta_3}\left(\sigma+(\sigma^*-\hat{\sigma})\right) \ =\ -\Delta_3\left(\sigma+(\sigma^*-\hat{\sigma})\right) \\
		& \ =\ U({\bf z_2, t_2}+\sigma^*-\hat{\sigma}+\sigma)-U({\bf z_1, t_1}+\sigma^*-\hat{\sigma}+\sigma),
	\end{align*}
	so that for the shifted function (see Figure \ref{Delta3tilde} in the Appendix):
	\[
	\tilde{\Delta}_3(\sigma_1)>0, \tilde{\Delta}_3(\hat{\sigma})=0, \tilde{\Delta}_3(\sigma_2)<0, \text{ and }\tilde{\Delta}_3(\sigma')<0.
	\]
	Define the mixtures $({\bf w_1, l_1}) = ({\bf z_2}, {\bf t_2}+\sigma^*-\hat{\sigma})\lambda({\bf x}, {\bf t})$ and $({\bf w_2, l_2}) = ({\bf z_1}, {\bf t_1}+\sigma^*-\hat{\sigma})\lambda({\bf y}, {\bf s})$.
	Analyse the function
	\begin{align*}
		\Delta_4(\sigma) \ =\ &U({\bf w_1, l_1}+\sigma)-U({\bf w_2, l_2}+\sigma)\\
		\ =\ &\lambda\left(U({\bf z_2}, {\bf t_2}+\sigma^*-\hat{\sigma}+\sigma)-U ({\bf z_1}, {\bf t_1}+\sigma^*-\hat{\sigma}+\sigma)\right)\\
		&+(1-\lambda)\left(U({\bf x}, {\bf t}+\sigma)-U({\bf y}, {\bf s}+\sigma) \right)\\
		\ =\ &\lambda \tilde{\Delta}_3(\sigma)+(1-\lambda) \Delta_1(\sigma).
	\end{align*}
	Choosing $\lambda$ sufficiently small we obtain (see Figure \ref{Delta4} in the Appendix):
	\[
	\Delta_4(\sigma_1)>0, \Delta_4(\sigma_2)<0, \text{ and } \Delta_4(\sigma')>0.
	\]
	Therefore, we have arrived at a contradiction to the one-switch property.
	
	{\em Case 2}. We next assume that there exist $\sigma^*_1<\sigma^*_2$ such that $\Delta_3(\sigma)=0$ if and only if $\sigma \in [\sigma^*_1, \sigma^*_2] \subset (\sigma_2, \sigma')$. Also, by the one-switch property, we must have $\Delta_3(\sigma)<0$ if $\sigma<\sigma_1^*$, and  $\Delta_3(\sigma)>0$ if $\sigma>\sigma_2^*$. We consider the reflection of $\Delta_3(\sigma)$ about the $\sigma$-axis and shift it to the left by a small amount $\epsilon>0$ to obtain $\bar{\Delta_{3}}(\sigma)$ as follows
	\[
	\bar{\Delta_3}(\sigma) \ =\ -\Delta_3(\sigma+\varepsilon) \ =\  U({\bf z_2, t_2}+\varepsilon+\sigma)-U({\bf z_1, t_1}+\varepsilon+\sigma).
	\]
	Then (see Figure \ref{Delta3bar} in the Appendix):
	\[
	\bar{\Delta_3}(\sigma)>0 \text{ if } \sigma<\sigma_1^*-\varepsilon, \bar{\Delta_3}(\sigma)=0 \text{ if } \sigma \in [\sigma_1^*-\varepsilon, \sigma_2^*-\varepsilon], \text { and } \bar{\Delta_3}(\sigma)<0 \text{ if } \sigma>\sigma_2^*-\varepsilon.
	\]
	Define the mixtures $({\bf v_1, k_1}) = ({\bf z_2}, {\bf t_2}+\varepsilon)\lambda({\bf z_1}, {\bf t_1})$ and $({\bf v_2, k_2}) = ({\bf z_1}, {\bf t_1}+\varepsilon)\lambda({\bf z_2}, {\bf t_2})$.
	Analyse the function
	\begin{align*}
		\bar{\Delta_4}(\sigma) \ =\  &U({\bf v_1, k_1}+\sigma)-U({\bf v_2, k_2}+\sigma)\\
		\ =\  &\lambda\left(U({\bf z_2}, {\bf t_2}+\varepsilon+\sigma)-U ({\bf z_1}, {\bf t_1}+\varepsilon+\sigma)\right)+\\
		&(1-\lambda)\left(U({\bf z_1}, {\bf t_1}+\sigma)-U({\bf z_2}, {\bf t_2}+\sigma) \right)\\
		\ =\  &\lambda \bar{\Delta}_3(\sigma)+(1-\lambda) \Delta_3(\sigma).
	\end{align*}
	Choosing $\lambda$ sufficiently small we obtain (see Figure \ref{Delta4bar} in the Appendix):
	\[
	\bar{\Delta_4}(\sigma_1)<0, \bar{\Delta_4}(\sigma_1^*)=0, \text{ and } \bar{\Delta_4}(\sigma_2^*)<0.
	\]
	Finally, mixing $(b, p)\lambda({\bf v_1}, {\bf k_1})$ and $(a, q)\lambda({\bf v_2}, {\bf k_2})$ we have 
	\begin{align*}
		\Delta_5(\sigma)& \ =\ \lambda \left(U(b, p+\sigma) -U(a, q+\sigma) \right)+(1-\lambda) \left(U({\bf v_1}, {\bf k_1}+\sigma) -U({\bf v_2}, {\bf k_2}+\sigma) \right)\\
		& \ =\ -\lambda \Delta_2(\sigma)+(1-\lambda)\bar{\Delta_4}(\sigma).
	\end{align*}
	Recall that $\Delta_2(\sigma)<0$ for all $\sigma$. Letting $\lambda$ be sufficiently small we obtain a double switch (see Figure \ref{Delta5} in the Appendix):
	\[
	\Delta_5(\sigma_1)<0, \Delta_5(\sigma_1^*)>0, \text{ and } \Delta_5(\sigma_2^*)<0,	
	\]
	which is a contradiction.
\end{proof}
The implication of Proposition \ref{alevels} is as follows:
\begin{corollary}
	If preferences $\succcurlyeq$ on $\mathcal{A}$ have a DEU representation, then the one-switch property is equivalent to the weak one-switch property.
\end{corollary}

\subsection{The weak one-switch property for dated outcomes and impatience}

In this section we consider preferences $\succcurlyeq$ on ${\mathcal A_1}$ and return to the initial assumption that $X=[0, \infty)$.
\begin{definition}
	We say that $\succcurlyeq$ exhibits the {\em weak one-switch property for dated outcomes} if $\succcurlyeq$ exhibits the weak one-switch property on $\mathcal A_1$.
\end{definition}
The following proposition gives a partial characterisation of the weak one-switch property for dated outcomes.
\begin{proposition}\label{DIimpliesSwitch}
	Let $\succcurlyeq$ restricted to ${\mathcal A_1}$ has a DU representation. Then preferences $\succcurlyeq$ exhibit the weak one-switch property for dated outcomes if $\succcurlyeq$ exhibit DI or II.
\end{proposition}
\begin{proof} We show that if $\succcurlyeq$ exhibits II or DI it must also exhibit the weak one-switch property for dated outcomes.
	The proof is by contradiction.
	Suppose that for some $(x,t), (y, s) \in  {\mathcal A_1}$ and some $\sigma, \varepsilon$ such that $0<\sigma< \varepsilon$ we can have: 
	\begin{align}
		\label{first} (x, t) &\succ (y, s), \\
		\label{second} (x, t+\sigma) &\prec (y, s+\sigma),\\
		\label{third} (x, t+\varepsilon) &\succ (y, s+\varepsilon).
	\end{align}
	W.l.o.g. assume that $t<s$. Then we also must have that $x<y$, otherwise \eqref{second} contradicts impatience and monotonicity.
	By continuity we can find $s'\in (t, s)$ such that $(x, t) \sim (y, s')$. Therefore, by II we have $(x, t+\sigma) \succcurlyeq (y, s'+\sigma)$. By impatience,  $(y, s'+\sigma) \succ (y, s+\sigma)$, hence, by transitivity,
	$(x, t+\sigma) \succ (y, s+\sigma)$, which is a contradiction to \eqref{second}. Therefore, we have shown that if $\succcurlyeq$ does not satisfy the one-switch property for dated outcomes it also does not exhibit II.
	
	The proof for DI is analogous. Indeed, consider $(x, t+\sigma) \prec (y, s+\sigma)$, with $t<s$, $x<y$, as before. By continuity we can find $s''>s$ such that $(x, t+\sigma) \sim (y, s''+\sigma)$. 
	Since $\varepsilon>\sigma$, let $\varepsilon=\sigma+\gamma$ for some $\gamma>0$. From the equivalence $(x, t+\sigma) \sim (y, s''+\sigma)$ it follows by DI that $(x, t+\sigma+\gamma) \preccurlyeq (y, s''+\sigma+\gamma)$, or, equivalently, $(x, t+\varepsilon) \preccurlyeq (y, s''+\varepsilon)$. Since $s''>s$, impatience implies that $(y, s''+\varepsilon)\prec (y, s+\varepsilon)$. Therefore, we must have $(x, t+\varepsilon) \prec (y, s+\varepsilon)$, which brings us to a contradiction to \eqref{third}.
	%
\end{proof}

We have not been able to establish the converse to Proposition \ref{DIimpliesSwitch}. The arguments used in Lemma \ref{1sismon} do not adapt straightforwardly to the present situation.
However, Proposition \ref{DIimpliesSwitch}, together with previous characterisation of the one-switch property for dated outcomes (Lemma \ref{1sismon}), already imply that the one-switch property for dated outcomes and the weak one-switch property for dated outcomes are not equivalent for intertemporal preferences with a DU representation.

\section{Discussion}\label{section:discOSP}

This paper fills a gap in Bell's original characterisation \cite[Proposition 8]{bell1988one} of discount functions compatible with the one-switch property for sequences of dated outcomes. We showed that functions of the linear times exponential form also have this property and that such discount functions exhibit strictly increasing impatience. Although decreasing impatience is commonly found in experiments \cite{frederick2002time}, there is also much evidence for increasing impatience (see, for example, \cite{attema2010time}, \cite{takeuchi2011non}). To the best of our knowledge, the linear times exponential function has not been used in the literature on time preferences before. Therefore, we introduce a new type of a discount function that accommodates strictly increasing impatience and the one-switch property. 

With regard to sums of exponentials, there has recently been some interest in this type of discount function. In their experiments, McClure et al. \cite{mcclure2007time} used magnetic neuro-images of individuals' brains to study intertemporal choice. They found that making a decision is a result of the interaction of two separate neural systems with different levels of impatience. To describe the involvement of these two brain areas in discounting, they suggested sum of exponential discount functions which they refer to as {\em double exponential} discounting. The recent empirical study by Cavagnaro et al. \cite{cavagnaro2016functional} demonstrates that double exponential discounting provides a better fit to actual time preferences than exponential, quasi-hyperbolic, proportional hyperbolic and generalized hyperbolic discounting.

A second contribution of this paper is to clarify the definition of the original one-switch rule introduced by Bell \cite{bell1988one}, by considering two definitions: the weak one-switch property and the (strong) one-switch property. We demonstrate that if $X$ is a mixture set and if preferences have a DEU representation, then the one-switch property is equivalent to the weak one-switch property.

Third, we prove that if preferences have a DU representation, then preferences exhibit the one-switch property for dated outcomes if and only if they exhibit either strictly increasing impatience, or strictly decreasing impatience, or constant impatience. A partial analogue is obtained for the weak one-switch property for dated outcomes. That is, the preferences exhibit the weak one-switch property for dated outcomes if they exhibit increasing impatience or decreasing impatience.

\section*{Acknowledgments} 
I thank my supervisor Matthew Ryan for valuable suggestions and detailed comments. I am grateful to participants of 6th Microeconomic Theory Workshop at Victoria University of Wellington for helpful comments. Financial support from the University of Auckland is gratefully acknowledged.

\appendix
\counterwithin{figure}{section}
\section{Appendix}
\subsection{Illustrations for the proof of Proposition \ref{alevels}}

In this appendix we provide the illustrations to accompany the proof of Proposition \ref{alevels}. Note that for all the figures in the appendix only the {\em sign} (but not the value) of a vertical coordinate of a point is relevant.

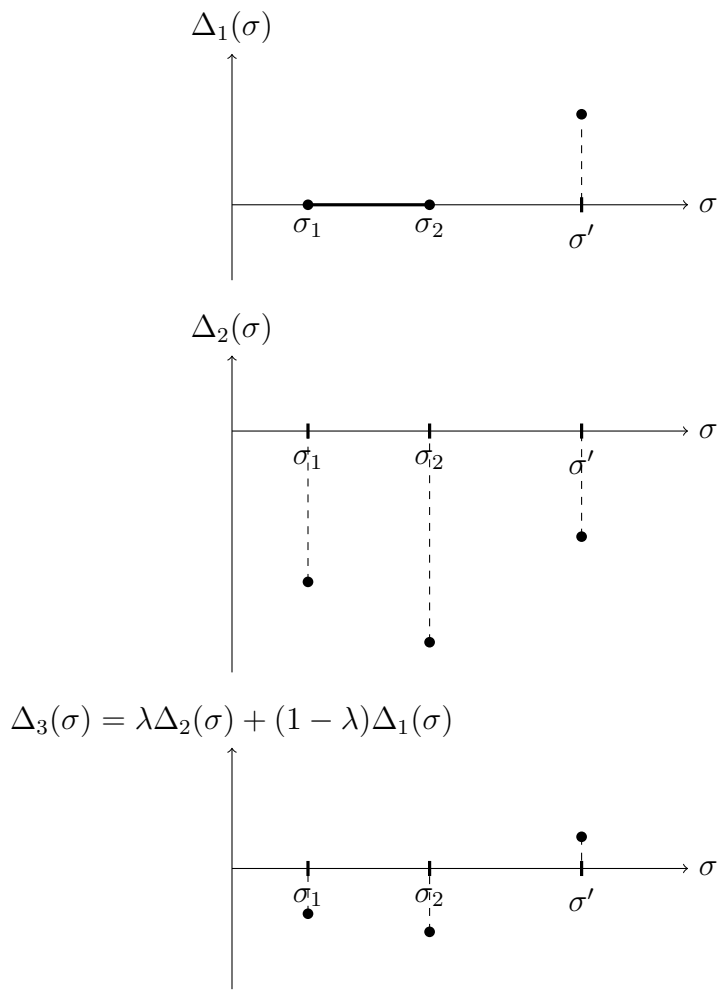
\begin{figure}[htb!] 
	\centering
	\begin{tikzpicture}[
	scale=2,
	every node/.style={color=black},
	dot/.style={circle,fill=black,minimum size=4pt,inner sep=0pt,
		outer sep=-1pt},
	]
	\draw[->] (0, 0)--(3,0) node(xline)[right] {$\sigma$}; 
	\draw[->] (0,-0.5)--(0,1) node(yline)[above] {$\Delta_1(\sigma)$};
	\node[dot,label=below:$\sigma_1$] at (0.5,0) {};
	\node[dot,label=below:$\sigma_2$] at (1.3,0) {};
	\node[dot] at (2.3,0.6) {};
	\draw [very thick] (2.3,-.05) node[below]{$\sigma'$} -- (2.3,0.05);
	\draw[dashed] (2.3,0.5)--(2.3,0);
	\draw [very thick] (0.5,0)--(1.3,0);
	\begin{scope}[yshift=-1.5cm]
	\draw[->] (0, 0)--(3,0) node(xline)[right] {$\sigma$}; 
	\draw[->] (0,-1.6)--(0,0.5) node(yline)[above] {$\Delta_2(\sigma)$};
	\node[dot] at (0.5,-1) {};
	\node[dot] at (1.3,-1.4) {};
	\node[dot] at (2.3,-0.7) {};
	\draw [very thick] (0.5,-.05) node[below]{$\sigma_1$} -- (0.5,0.05);
	\draw [very thick] (1.3,-.05) node[below]{$\sigma_2$} -- (1.3,0.05);
	\draw [very thick] (2.3,-.05) node[below]{$\sigma'$} -- (2.3,0.05);
	\draw[dashed] (0.5,-1)--(0.5,0);
	\draw[dashed] (1.3,-1.4)--(1.3,0);
	\draw[dashed] (2.3,-0.7)--(2.3,0);
	\end{scope}
	\begin{scope}[yshift=-4.4cm]
	\draw[->] (0, 0)--(3,0) node(xline)[right] {$\sigma$}; 
	\draw[->] (0,-0.8)--(0,0.8) node(yline)[above] {$\Delta_3(\sigma)=\lambda \Delta_2(\sigma)+(1-\lambda)\Delta_1(\sigma)$};
	\node[dot] at (0.5,-0.3)  {};
	\node[dot] at (1.3,-0.42)  {};
	\node[dot] at (2.3,0.21)  {};
	\draw [very thick] (0.5,-.05) node[below]{$\sigma_1$} -- (0.5,0.05);
	\draw [very thick] (1.3,-.05) node[below]{$\sigma_2$} -- (1.3,0.05);
	\draw [very thick] (2.3,-.05) node[below]{$\sigma'$} -- (2.3,0.05);
	\draw[dashed] (0.5,-0.3)--(0.5,0);
	\draw[dashed] (1.3,-0.42)--(1.3,0);
	\draw[dashed] (2.3,0.21)--(2.3,0);
	\end{scope}
	\end{tikzpicture}
	\caption{The mixture of $\Delta_2(\sigma)$ and $\Delta_1(\sigma)$ (with small $\lambda$)}
	\label{Delta3}
\end{figure}	
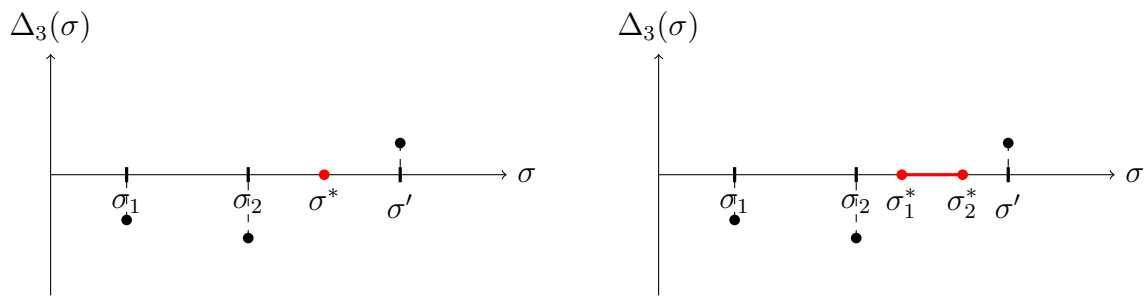
\begin{figure}[htb!] 
	\begin{tikzpicture}[
	scale=2,
	every node/.style={color=black},
	dot/.style={circle,fill=black,minimum size=4pt,inner sep=0pt,
		outer sep=-1pt},
	]
	\draw[->] (0, 0)--(3,0) node(xline)[right] {$\sigma$}; 
	\draw[->] (0,-0.8)--(0,0.8) node(yline)[above] {$\Delta_3(\sigma)$};
	\node[dot] at (0.5,-0.3)  {};
	\node[dot] at (1.3,-0.42)  {};
	\node[dot] at (2.3,0.21)  {};
	\node[dot,red, label=below:$\sigma^*$] at (1.8,0)  {};
	\draw [very thick] (0.5,-.05) node[below]{$\sigma_1$} -- (0.5,0.05);
	\draw [very thick] (1.3,-.05) node[below]{$\sigma_2$} -- (1.3,0.05);
	\draw [very thick] (2.3,-.05) node[below]{$\sigma'$} -- (2.3,0.05);
	\draw[dashed] (0.5,-0.3)--(0.5,0);
	\draw[dashed] (1.3,-0.42)--(1.3,0);
	\draw[dashed] (2.3,0.21)--(2.3,0);
	
	\begin{scope}[xshift=4cm]
	\draw[->] (0, 0)--(3,0) node(xline)[right] {$\sigma$}; 
	\draw[->] (0,-0.8)--(0,0.8) node(yline)[above] {$\Delta_3(\sigma)$};
	\node[dot] at (0.5,-0.3)  {};
	\node[dot] at (1.3,-0.42)  {};
	\node[dot,red, label=below:$\sigma_1^*$] at (1.6,0)  {};
	\node[dot,red, label=below:$\sigma_2^*$] at (2.0,0)  {};
	\node[dot] at (2.3,0.21)  {};
	\draw [very thick] (0.5,-.05) node[below]{$\sigma_1$} -- (0.5,0.05);
	\draw [very thick] (1.3,-.05) node[below]{$\sigma_2$} -- (1.3,0.05);
	\draw [very thick] (2.3,-.05) node[below]{$\sigma'$} -- (2.3,0.05);
	\draw[dashed] (0.5,-0.3)--(0.5,0);
	\draw[dashed] (1.3,-0.42)--(1.3,0);
	\draw[dashed] (2.3,0.21)--(2.3,0);
	\draw [very thick, red] (1.6,0)--(2.0,0);
	\end{scope}
	\end{tikzpicture}	
	\caption{Function $\Delta_3(\sigma)$ equals zero at one point or at the interval}
	\label{Two_options}
\end{figure}

\clearpage
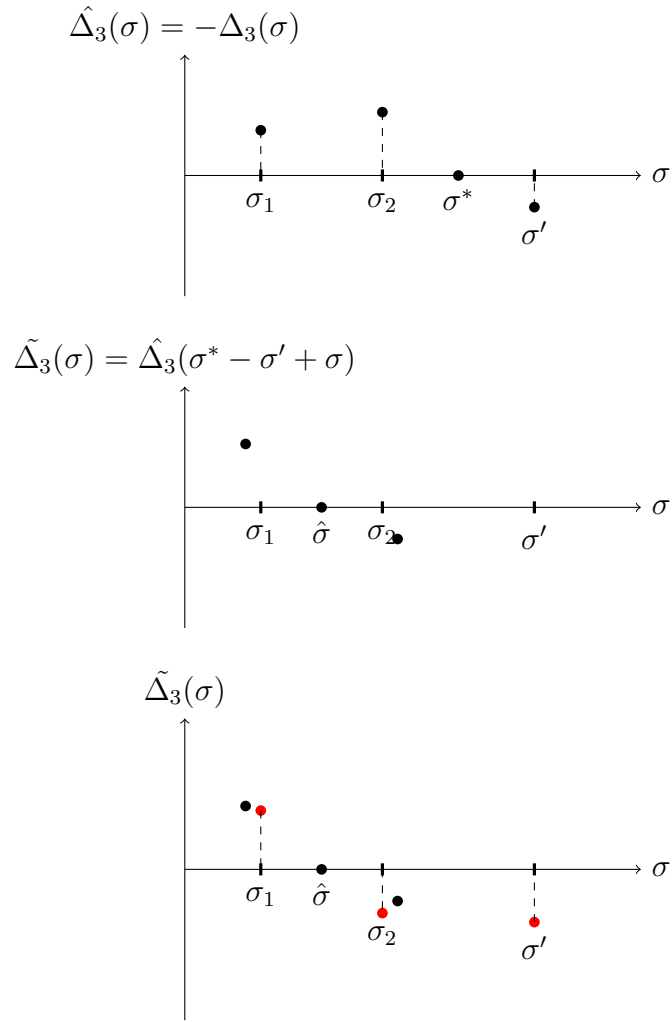
\begin{figure}[htb!] 
	\centering
	\begin{tikzpicture}[
	scale=2,
	every node/.style={color=black},
	dot/.style={circle,fill=black,minimum size=4pt,inner sep=0pt,
		outer sep=-1pt},
	]
	
	\draw[->] (0, 0)--(3,0) node(xline)[right] {$\sigma$}; 
	\draw[->] (0,-0.8)--(0,0.8) node(yline)[above] {$\hat{\Delta_3}(\sigma)=-\Delta_3(\sigma)$};
	\node[dot] at (0.5,0.3) (sigma1) {};
	\node[dot] at (1.3,0.42)  {};
	\node[dot,label=below:$\sigma^*$] at (1.8,0) (sigstar) {};
	\node[dot, label=below:$\sigma'$] at (2.3,-0.21)  {};
	\draw [very thick] (0.5,-.04) node[below]{$\sigma_1$} -- (0.5,0.04);
	\draw [very thick] (1.3,-.04) node[below]{$\sigma_2$} -- (1.3,0.04);
	\draw [very thick] (2.3,-.04) -- (2.3,0.04);
	\draw[dashed] (0.5,0.3)--(0.5,0);
	\draw[dashed] (1.3,0.42)--(1.3,0);
	\draw[dashed] (2.3,-0.21)--(2.3,0);
	
	\begin{scope}[yshift=-2.2cm]
	\draw[->] (0, 0)--(3,0) node(xline)[right] {$\sigma$}; 
	\draw[->] (0,-0.8)--(0,0.8) node(yline)[above] {$\tilde{\Delta_3}(\sigma)=\hat{\Delta_3}(\sigma^*-\sigma'+\sigma)$};
	\node[dot] at (0.4,0.42)  {}; 
	\node[dot, label=below:$\hat{\sigma}$] at (0.9,0) (sigstar) {};
	\node[dot] at (1.4,-0.21)  {}; 
	\draw [very thick] (0.5,-.04) node[below]{$\sigma_1$} -- (0.5,0.04);
	\draw [very thick] (1.3,-.04) node[below]{$\sigma_2$} -- (1.3,0.04);
	\draw [very thick] (2.3,-.04) node[below]{$\sigma'$} -- (2.3,0.04);
	
	\end{scope}
	
	\begin{scope}[yshift=-4.6cm]
	\draw[->] (0, 0)--(3,0) node(xline)[right] {$\sigma$}; 
	\draw[->] (0,-1)--(0,1) node(yline)[above] {$\tilde{\Delta_3}(\sigma)$};
	\node[dot] at (0.4,0.42)  {}; 
	\node[dot, red] at (0.5,0.39) {}; 
	\node[dot, label=below:$\hat{\sigma}$] at (0.9,0) (sigstar) {};
	\node[dot, red, label=below:$\sigma_2$] at (1.3,-0.29)  {}; 
	\node[dot] at (1.4,-0.21)  {}; 
	\node[dot, red,label=below:$\sigma'$] at (2.3,-0.35)  {}; 
	\draw [very thick] (0.5,-.04) node[below]{$\sigma_1$} -- (0.5,0.04);
	\draw [very thick] (1.3,-.04) -- (1.3,0.04);
	\draw [very thick] (2.3,-.04) -- (2.3,0.04);
	\draw[dashed] (0.5,0.39)--(0.5,0);
	\draw[dashed] (1.3,0)--(1.3,-0.29);
	\draw[dashed] (2.3,-0.35)--(2.3,0);
	
	\end{scope}
	\end{tikzpicture}
	\caption{Transformation of $\Delta_3(\sigma)$ to $\tilde{\Delta_3}(\sigma)$, and then with the new values at indicated $\sigma_1, \sigma_2, \sigma'$}	
	\label{Delta3tilde}
\end{figure}

\newpage
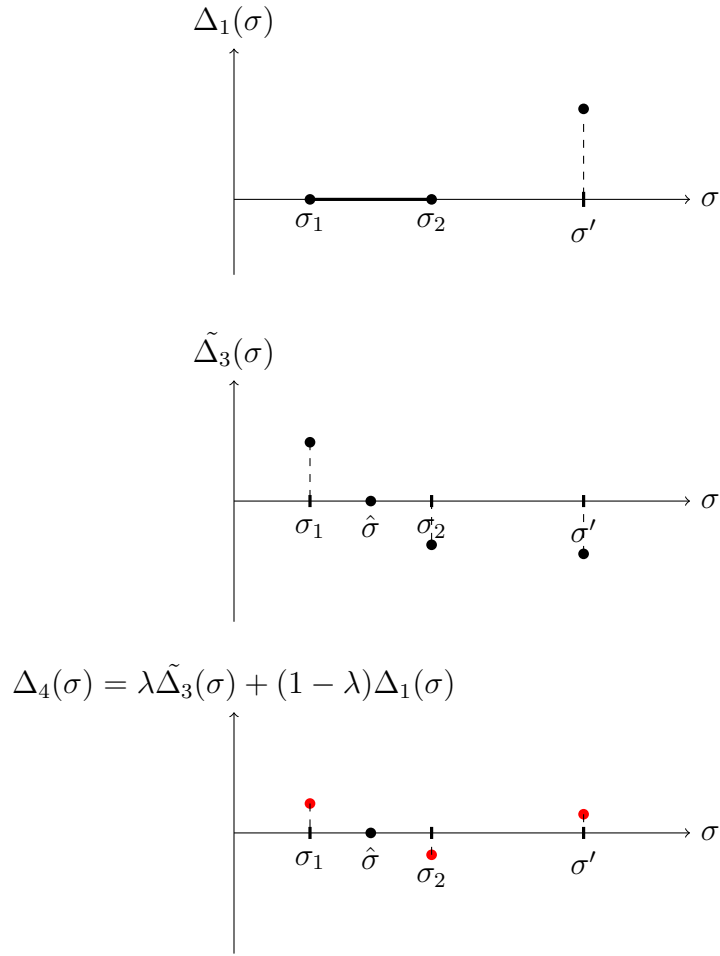
\begin{figure}[htb!]
	\centering
	\begin{tikzpicture}[
	scale=2,
	every node/.style={color=black},
	dot/.style={circle,fill=black,minimum size=4pt,inner sep=0pt,
		outer sep=-1pt},
	]
	
	\draw[->] (0, 0)--(3,0) node(xline)[right] {$\sigma$}; 
	\draw[->] (0,-0.5)--(0,1) node(yline)[above] {$\Delta_1(\sigma)$};
	\node[dot,label=below:$\sigma_1$] at (0.5,0)  {};
	\node[dot,label=below:$\sigma_2$] at (1.3,0)  {};
	\node[dot] at (2.3,0.6)  {};
	\draw [very thick] (2.3,-.05) node[below]{$\sigma'$} -- (2.3,0.05);
	\draw[dashed] (2.3,0.5)--(2.3,0);
	\draw [very thick] (0.5,0)--(1.3,0);	
	
	\begin{scope} [yshift=-2cm]
	\draw[->] (0, 0)--(3,0) node(xline)[right] {$\sigma$}; 
	\draw[->] (0,-0.8)--(0,0.8) node(yline)[above] {$\tilde{\Delta_3}(\sigma)$};
	\node[dot] at (0.5,0.39) {}; 
	\node[dot, label=below:$\hat{\sigma}$] at (0.9,0) (sigstar) {};
	\node[dot] at (1.3,-0.29)  {}; 
	\node[dot] at (2.3,-0.35)  {}; 
	\draw [very thick] (0.5,-.04) node[below]{$\sigma_1$} -- (0.5,0.04);
	\draw [very thick] (1.3,-.04) node[below]{$\sigma_2$} -- (1.3,0.04);
	\draw [very thick] (2.3,-.04) node[below]{$\sigma'$} -- (2.3,0.04);
	\draw[dashed] (0.5,0.39)--(0.5,0);
	\draw[dashed] (1.3,0)--(1.3,-0.29);
	\draw[dashed] (2.3,-0.35)--(2.3,0);
	
	\end{scope}
	
	\begin{scope}[yshift=-4.2cm]
	\draw[->] (0, 0)--(3,0) node(xline)[right] {$\sigma$}; 
	\draw[->] (0,-0.8)--(0,0.8) node(yline)[above] {$\Delta_4(\sigma)=\lambda\tilde{\Delta_3}(\sigma)+(1-\lambda)\Delta_1(\sigma)$};
	\node[dot, red] at (0.5,0.195) (sigma1) {};
	\node[dot, label=below:$\hat{\sigma}$] at (0.9,0) (sigstar) {};
	\node[dot, red,label=below:$\sigma_2$] at (1.3,-0.145)  {};
	\node[dot, red] at (2.3,0.124)  {};
	\draw [very thick] (0.5,-.04) node[below]{$\sigma_1$} -- (0.5,0.04);
	\draw [very thick] (1.3,-.04) -- (1.3,0.04);
	\draw [very thick] (2.3,-.04) node[below]{$\sigma'$} -- (2.3,0.04);
	\draw[dashed] (0.5,0.195)--(0.5,0);
	\draw[dashed] (1.3,-0.145)--(1.3,0);
	\draw[dashed] (2.3,0.124)--(2.3,0);
	
	\end{scope}
	\end{tikzpicture}
	\caption{The mixture of $\tilde{\Delta_3}(\sigma)$ and $\Delta_1(\sigma)$ (with small $\lambda$) produces a double switch}	
	\label{Delta4}
\end{figure}	

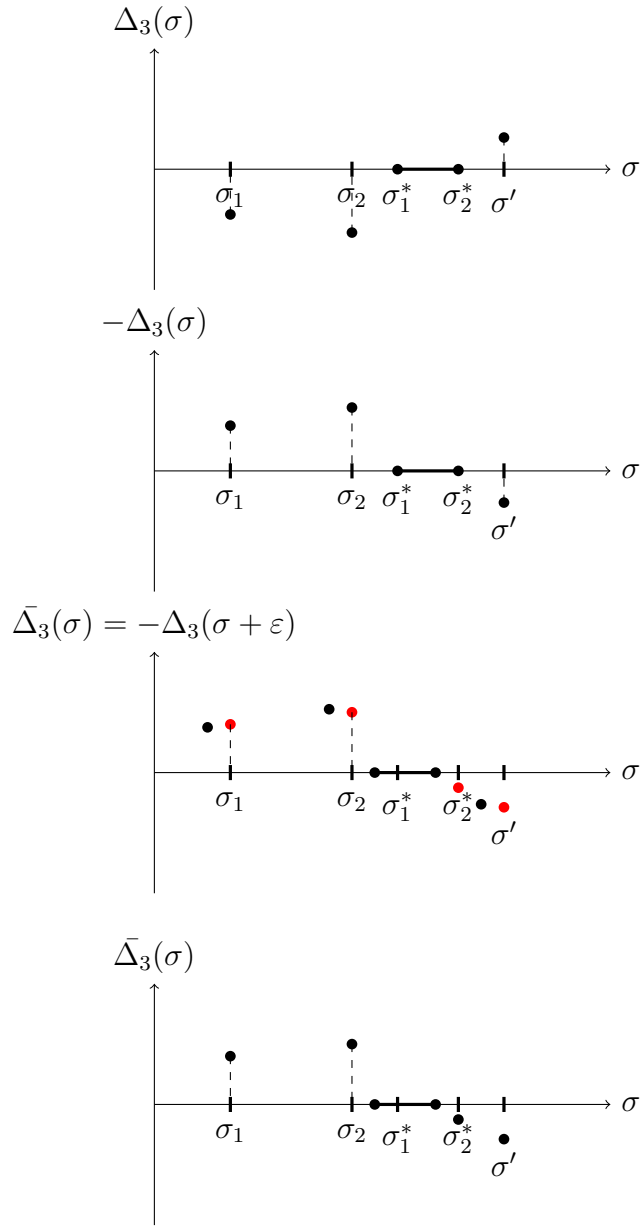
\begin{figure}[H] 
	\centering
	\begin{tikzpicture}[
	scale=2,
	every node/.style={color=black},
	dot/.style={circle,fill=black,minimum size=4pt,inner sep=0pt,
		outer sep=-1pt},
	]
	\draw[->] (0, 0)--(3,0) node(xline)[right] {$\sigma$}; 
	\draw[->] (0,-0.8)--(0,0.8) node(yline)[above] {$\Delta_3(\sigma)$};
	\node[dot] at (0.5,-0.3)  {};
	\node[dot] at (1.3,-0.42)  {};
	\node[dot, label=below:$\sigma_1^*$] at (1.6,0)  {};
	\node[dot, label=below:$\sigma_2^*$] at (2.0,0)  {};
	\node[dot] at (2.3,0.21)  {};
	\draw [very thick] (0.5,-.05) node[below]{$\sigma_1$} -- (0.5,0.05);
	\draw [very thick] (1.3,-.05) node[below]{$\sigma_2$} -- (1.3,0.05);
	\draw [very thick] (2.3,-.05) node[below]{$\sigma'$} -- (2.3,0.05);
	\draw[dashed] (0.5,-0.3)--(0.5,0);
	\draw[dashed] (1.3,-0.42)--(1.3,0);
	\draw[dashed] (2.3,0.21)--(2.3,0);
	\draw [very thick] (1.6,0)--(2.0,0);
	
	\begin{scope} [yshift=-2cm]
	\draw[->] (0, 0)--(3,0) node(xline)[right] {$\sigma$}; 
	\draw[->] (0,-0.8)--(0,0.8) node(yline)[above] {$-\Delta_3(\sigma)$};
	\node[dot] at (0.5,0.3)  {};
	\node[dot] at (1.3,0.42)  {};
	\node[dot, label=below:$\sigma_1^*$] at (1.6,0)  {};
	\node[dot, label=below:$\sigma_2^*$] at (2.0,0)  {};
	\node[dot, label=below:$\sigma'$] at (2.3,-0.21)  {};
	\draw [very thick] (0.5,-.05) node[below]{$\sigma_1$} -- (0.5,0.05);
	\draw [very thick] (1.3,-.05) node[below]{$\sigma_2$} -- (1.3,0.05);
	\draw [very thick] (2.3,-.05) -- (2.3,0.05);
	\draw[dashed] (0.5,0.3)--(0.5,0);
	\draw[dashed] (1.3,0.42)--(1.3,0);
	\draw[dashed] (2.3,-0.21)--(2.3,0);
	\draw [very thick] (1.6,0)--(2.0,0);
	\end{scope}	
	
	\begin{scope} [yshift=-4cm]
	\draw[->] (0, 0)--(3,0) node(xline)[right] {$\sigma$}; 
	\draw[->] (0,-0.8)--(0,0.8) node(yline)[above] {$\bar{\Delta_3}(\sigma)=-\Delta_3(\sigma+\varepsilon)$};
	\node[dot] at (0.35,0.3)  {}; 
	\node[dot, red] at (0.5,0.32)  {}; 
	\node[dot] at (1.15,0.42)  {};
	\node[dot, red] at (1.3,0.4)  {};
	\node[dot] at (2.15,-0.21)  {};
	\node[dot] at (1.45,0)  {};
	\node[dot] at (1.85,0)  {};
	\node[dot, red] at (2,-0.1)  {};
	\node[dot, label=below:$\sigma'$, red] at (2.3,-0.23)  {};
	\draw [very thick] (0.5,-.05) node[below]{$\sigma_1$} -- (0.5,0.05);
	\draw [very thick] (1.3,-.05) node[below]{$\sigma_2$} -- (1.3,0.05);
	\draw [very thick] (1.6,-.05) node[below]{$\sigma_1^*$} -- (1.6,0.05);
	\draw [very thick] (2,-.05) node[below]{$\sigma_2^*$}-- (2,0.05);
	\draw [very thick] (2.3,-.05) -- (2.3,0.05);		
	\draw[dashed] (0.5,0.32)--(0.5,0);	
	\draw[dashed] (1.3,0.4)--(1.3,0);
	\draw [very thick] (1.45,0)--(1.85,0);
	\end{scope}	
	
	\begin{scope} [yshift=-6.2cm]
	\draw[->] (0, 0)--(3,0) node(xline)[right] {$\sigma$}; 
	\draw[->] (0,-0.8)--(0,0.8) node(yline)[above] {$\bar{\Delta_3}(\sigma)$};
	\node[dot] at (0.5,0.32)  {}; 
	
	\node[dot] at (1.3,0.4)  {};
	\node[dot] at (1.45,0)  {};
	\node[dot] at (1.85,0)  {};
	\node[dot] at (2,-0.1)  {};
	\node[dot, label=below:$\sigma'$] at (2.3,-0.23)  {};
	\draw [very thick] (0.5,-.05) node[below]{$\sigma_1$} -- (0.5,0.05);
	\draw [very thick] (1.3,-.05) node[below]{$\sigma_2$} -- (1.3,0.05);
	\draw [very thick] (1.6,-.05) node[below]{$\sigma_1^*$} -- (1.6,0.05);
	\draw [very thick] (2,-.05) node[below]{$\sigma_2^*$}-- (2,0.05);
	\draw [very thick] (2.3,-.05) -- (2.3,0.05);
	\draw[dashed] (0.5,0.32)--(0.5,0);
	\draw[dashed] (1.3,0.4)--(1.3,0);
	\draw [very thick] (1.45,0)--(1.85,0);
	\end{scope}	
	\end{tikzpicture}
	\caption{Transformation of $\Delta_3(\sigma)$ to $\bar{\Delta_3}(\sigma)$  with the new values at $\sigma_1, \sigma_2, \sigma_2^*, \sigma'$}
	\label{Delta3bar}
\end{figure}

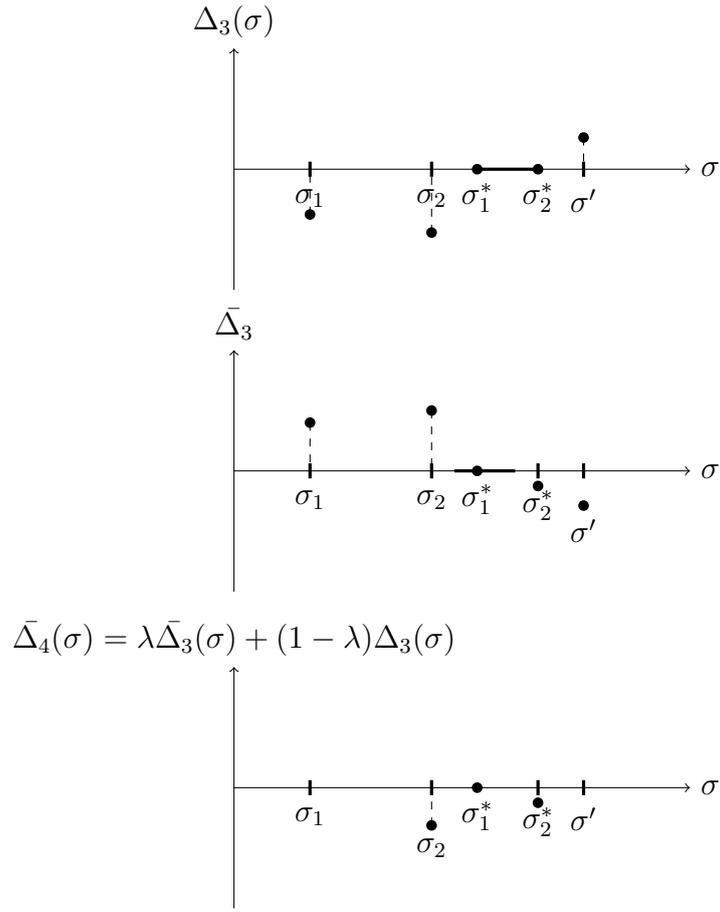
\begin{figure}[H] 
	\centering
	\begin{tikzpicture}[
	scale=2,
	every node/.style={color=black},
	dot/.style={circle,fill=black,minimum size=4pt,inner sep=0pt,
		outer sep=-1pt},
	]
	\draw[->] (0, 0)--(3,0) node(xline)[right] {$\sigma$}; 
	\draw[->] (0,-0.8)--(0,0.8) node(yline)[above] {$\Delta_3(\sigma)$};
	\node[dot] at (0.5,-0.3)  {};
	\node[dot] at (1.3,-0.42)  {};
	\node[dot, label=below:$\sigma_1^*$] at (1.6,0)  {};
	\node[dot, label=below:$\sigma_2^*$] at (2.0,0)  {};
	\node[dot] at (2.3,0.21)  {};
	\draw [very thick] (0.5,-.05) node[below]{$\sigma_1$} -- (0.5,0.05);
	\draw [very thick] (1.3,-.05) node[below]{$\sigma_2$} -- (1.3,0.05);		
	\draw [very thick] (2.3,-.05) node[below]{$\sigma'$} -- (2.3,0.05);
	\draw[dashed] (0.5,-0.3)--(0.5,0);
	\draw[dashed] (1.3,-0.42)--(1.3,0);
	\draw[dashed] (2.3,0.21)--(2.3,0);
	\draw [very thick] (1.6,0)--(2.0,0);
	
	\begin{scope}[yshift=-2cm]
	\draw[->] (0, 0)--(3,0) node(xline)[right] {$\sigma$}; 
	\draw[->] (0,-0.8)--(0,0.8) node(yline)[above] {$\bar{\Delta_3}$};
	\node[dot] at (0.5,0.32)  {};
	\node[dot] at (1.3,0.4)  {};
	\node[dot, label=below:$\sigma_1^*$ ] at (1.6,0)  {};
	\node[dot] at (2,-0.1)  {};
	\node[dot, label=below:$\sigma'$] at (2.3,-0.23)  {};
	\draw [very thick] (0.5,-.05) node[below]{$\sigma_1$} -- (0.5,0.05);
	\draw [very thick] (1.3,-.05) node[below]{$\sigma_2$} -- (1.3,0.05);
	\draw [very thick] (2,-.05) node[below]{$\sigma_2^*$}-- (2,0.05);
	\draw [very thick] (2.3,-.05) -- (2.3,0.05);
	\draw[dashed] (0.5,0.32)--(0.5,0);
	\draw[dashed] (1.3,0.4)--(1.3,0);
	\draw [very thick] (1.45,0)--(1.85,0);
	\end{scope}
	
	\begin{scope}[yshift=-4.1cm]
	\draw[->] (0, 0)--(3,0) node(xline)[right] {$\sigma$}; 
	\draw[->] (0,-0.8)--(0,0.8) node(yline)[above] {$\bar{\Delta_4}(\sigma)=\lambda\bar{\Delta_3}(\sigma)+(1-\lambda)\Delta_3(\sigma)$};
	\node[dot, label=below:$\sigma_2$] at (1.3,-0.25)  {};
	\node[dot, label=below:$\sigma_1^*$] at (1.6,0)  {};
	\node[dot] at (2,-0.1)  {};
	\draw [very thick] (0.5,-.05) node[below]{$\sigma_1$} -- (0.5,0.05);
	\draw [very thick] (1.3,-.05) -- (1.3,0.05);
	\draw [very thick] (2,-.05) node[below]{$\sigma_2^*$}-- (2,0.05);
	\draw [very thick] (2.3,-.05)node[below]{$\sigma'$}-- (2.3,0.05);
	\draw[dashed] (1.3,-0.25)--(1.3,0);
	\end{scope}
	\end{tikzpicture}		
	\caption{The mixture of $\bar{\Delta_3}(\sigma)$ and $\Delta_3(\sigma)$ (with small $\lambda$) at the three points $\sigma_2, \sigma_1^*, \sigma_2^*$}
	\label{Delta4bar}		
\end{figure}

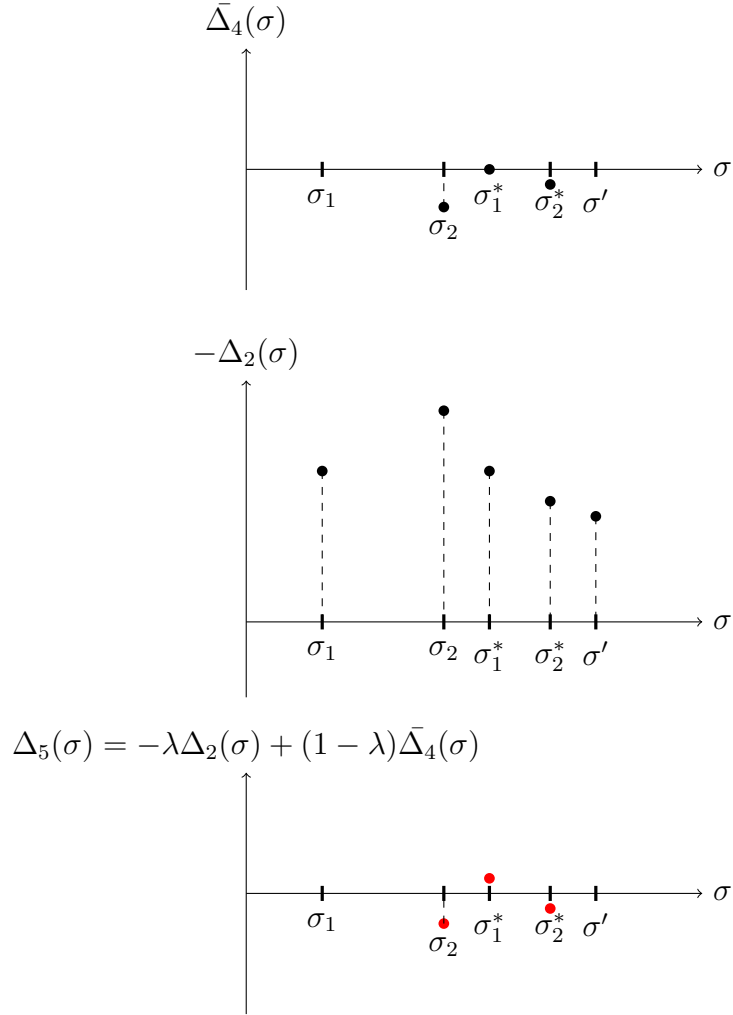
\begin{figure}[H] 
	\centering
	\begin{tikzpicture}[
	scale=2,
	every node/.style={color=black},
	dot/.style={circle,fill=black,minimum size=4pt,inner sep=0pt,
		outer sep=-1pt},
	]
	\draw[->] (0, 0)--(3,0) node(xline)[right] {$\sigma$}; 
	\draw[->] (0,-0.8)--(0,0.8) node(yline)[above] {$\bar{\Delta_4}(\sigma)$};
	\node[dot, label=below:$\sigma_2$] at (1.3,-0.25)  {};
	\node[dot, label=below:$\sigma_1^*$] at (1.6,0)  {};
	\node[dot] at (2,-0.1)  {};
	\draw [very thick] (0.5,-.05) node[below]{$\sigma_1$} -- (0.5,0.05);
	\draw [very thick] (1.3,-.05) -- (1.3,0.05);
	\draw [very thick] (2,-.05) node[below]{$\sigma_2^*$}-- (2,0.05);
	\draw [very thick] (2.3,-.05)node[below]{$\sigma'$}-- (2.3,0.05);
	\draw[dashed] (1.3,-0.25)--(1.3,0);
	
	\begin{scope}[yshift=-3cm]
	\draw[->] (0, 0)--(3,0) node(xline)[right] {$\sigma$}; 
	\draw[->] (0,-0.5)--(0,1.6) node(yline)[above] {$-\Delta_2(\sigma)$};
	\node[dot] at (0.5,1)  {};
	\node[dot] at (1.3,1.4)  {};	
	\node[dot] at (1.6,1)  {};
	\node[dot] at (2,0.8)  {};
	\node[dot] at (2.3,0.7)  {};
	\draw [very thick] (0.5,-.05) node[below]{$\sigma_1$} -- (0.5,0.05);
	\draw [very thick] (1.3,-.05) node[below]{$\sigma_2$} -- (1.3,0.05);
	\draw [very thick] (1.6,-.05) node[below]{$\sigma_1^*$} -- (1.6,0.05);
	\draw [very thick] (2,-.05) node[below]{$\sigma_2^*$} -- (2,0.05);
	\draw [very thick] (2.3,-.05) node[below]{$\sigma'$} -- (2.3,0.05);
	\draw[dashed] (0.5,1)--(0.5,0);	
	\draw[dashed] (1.3,1.4)--(1.3,0);
	\draw[dashed] (1.6,1)--(1.6,0);
	\draw[dashed] (2,0.8)--(2,0);
	\draw[dashed] (2.3,0.7)--(2.3,0);
	\end{scope}
	
	\begin{scope}[yshift=-4.8cm]
	\draw[->] (0, 0)--(3,0) node(xline)[right] {$\sigma$}; 
	\draw[->] (0,-0.8)--(0,0.8) node(yline)[above] {$\Delta_5(\sigma)=-\lambda\Delta_2(\sigma)+(1-\lambda)\bar{\Delta_4}(\sigma)$};
	\node[dot, label=below:$\sigma_2$, red] at (1.3,-0.2) {};
	\node[dot, red] at (1.6,0.1) {};
	\node[dot, red] at (2,-0.1) {};
	\draw [very thick] (0.5,-.05) node[below]{$\sigma_1$} -- (0.5,0.05);
	\draw [very thick] (1.3,-.05) -- (1.3,0.05);
	\draw [very thick] (1.6,-.05) node[below]{$\sigma_1^*$}-- (1.6,0.05);
	\draw [very thick] (2,-.05) node[below]{$\sigma_2^*$}-- (2,0.05);
	\draw [very thick] (2.3,-.05)node[below]{$\sigma'$}-- (2.3,0.05);
	\draw[dashed] (1.3,-0.2)--(1.3,0);
	\end{scope}
	\end{tikzpicture}		
	\caption{The mixture of $-\Delta_2(\sigma)$ and $\bar{\Delta_4}(\sigma)$ (with small $\lambda$) produces a double switch. It is sufficient to consider the value of $\Delta_5(\sigma)$ at the three points $\sigma_2, \sigma_1^*, \sigma_2^*$}
	\label{Delta5}		
\end{figure}

\bibliographystyle{plain}
\bibliography{Biblio}

\end{document}